\newif\ifarxiv
\arxivtrue

\ifarxiv
\documentclass[twocolumn,a4wide]{article}
\usepackage[
left=1.50cm,
right=1.50cm,
top=2.00cm,
bottom=2.00cm
]{geometry}
\usepackage{abstract,titlesec} 
\else
\documentclass[sigplan,screen]{acmart}
\fi
\usepackage[utf8]{inputenc}
\usepackage[group-separator={,}]{siunitx}
\usepackage{booktabs}
\usepackage[bottom]{footmisc}
\usepackage{hyperref}
\usepackage{graphicx}
\usepackage{float}
\usepackage{xcolor}
\usepackage{amsthm}
\theoremstyle{definition}

\newtheorem{theorem}{Theorem}[section]

\newtheorem{lemma}[theorem]{Lemma}

\newcommand{\ackname}{Acknowledgements}

\AtBeginDocument{%
  \providecommand\BibTeX{{%
    \normalfont B\kern-0.5em{\scshape i\kern-0.25em b}\kern-0.8em\TeX}}}

\begin{document}
\pagestyle{plain}

\title{Blockchain is Watching You: Profiling and Deanonymizing Ethereum Users\thanks{Support from the project 2018-1.2.1-NKP-00008:  Exploring the Mathematical Foundations of Artificial Intelligence of the Hungarian Government.}}
\ifarxiv

\author{Ferenc B\'eres$^{1,2}$, István A. Seres$^2$, Andr\'as A. Bencz\'ur$^{1,3}$, Mikerah Quintyne-Collins$^4$\\
 $^1$Institute for Computer Science and Control (SZTAKI), Budapest, Hungary\\
$^2$E\"otv\"os Loránd University, Budapest, Hungary\\
$^3$Széchenyi University, Gy\H{o}r, Hungary\\
$^4$HashCloack Inc., Toronto, Canada
}
\maketitle
\else
\author{}
\fi

\begin{abstract}
Ethereum is the largest public blockchain by usage. It applies an account-based model, which is inferior to Bitcoin's unspent transaction output model from a privacy perspective. Due to its privacy shortcomings, recently several privacy-enhancing overlays have been deployed on Ethereum, such as non-custodial, trustless coin mixers and confidential transactions.\\In our privacy analysis of Ethereum's account-based model, we describe several patterns that characterize only a limited set of users and successfully apply these ``quasi-identifiers'' in address deanonymization tasks. Using Ethereum Name Service identifiers as ground truth information, we quantitatively compare algorithms in recent branch of machine learning, the so-called graph representation learning, as well as time-of-day activity and transaction fee based user profiling techniques. As an application, we rigorously assess the privacy guarantees of the Tornado Cash coin mixer by discovering strong heuristics to link the mixing parties. To the best of our knowledge, we are the first to propose and implement Ethereum user profiling techniques based on quasi-identifiers.\\Finally, we describe a malicious value-fingerprinting attack, a variant of the Danaan-gift attack, applicable for the confidential transaction overlays on Ethereum. By incorporating user activity statistics from our data set, we estimate the success probability of such an attack.
\end{abstract}{}

\ifarxiv
\else
\maketitle
\fi

\section{Introduction}
The narrative around cryptocurrency privacy provisions has dramatically changed since the inception of Bitcoin~\cite{nakamoto2019bitcoin}. Initially many, especially criminals, thought Bitcoin and other cryptocurrencies provide privacy to hide their illicit business activities~\cite{christin2013traveling}. The first extensive study about Bitcoin's privacy provisions was done by Meiklejohn et al~\cite{meiklejohn2013fistful}, in which they provide several powerful heuristics allowing one to cluster Bitcoin addresses. The revelation of Bitcoin's privacy shortcomings spurred the creation and implementation of many privacy-enhancing overlays for Bitcoin~\cite{valenta2015blindcoin,bonneau2014mixcoin,ruffing2014coinshuffle,ziegeldorf2015coinparty}. As of today, several Bitcoin wallets, e.g. Wasabi and Samourai wallets, provide privacy-enhancing solutions to their users.

Previous work has focused on assessing the privacy guarantees provided by several UTXO-based (unspent transaction output) cryptocurrencies, such as Bitcoin~\cite{androulaki2013evaluating,meiklejohn2013fistful}, Monero~\cite{chervinski2019floodxmr,moser2018empirical,biryukov2019security} or Zcash~\cite{biryukov2018deanonymization,biryukov2019privacy,biryukov2019security,kappos2018empirical,tramer2019ping}.

However, perhaps surprisingly, until today there were no similar empirical studies on account-based cryptocurrency privacy provisions. Therefore in this work, we put forth the problem of studying the privacy guarantees of Ethereum's account-based model. Assessing and understanding the privacy guarantees of cryptocurrencies is essential as the lack of financial privacy is detrimental to most cryptocurrency use cases. Furthermore, there are state-sponsored companies and other entities, e.g. Chainalysis~\cite{nelson2020inside}, performing large-scale deanonymization tasks on cryptocurrency users. 

In contrast to the UTXO-model, many cryptocurrencies that provide smart contract functionalities operate with accounts. In an account-based cryptocurrency, users store their assets in accounts rather than in UTXOs. Already in the Bitcoin white paper, Nakamoto suggested that ``a new key pair should be used for each transaction to keep them from being linked to a common owner''~\cite{nakamoto2019bitcoin}. Despite this suggestion, account-based cryptocurrency users tend to use only a handful of addresses for their activities. In an account-based cryptocurrency, native transactions can only move funds between a single sender and a single receiver, hence in a payment transaction, the change remains at the sender account. Thus, a subsequent transaction necessarily uses the same address again to spend the remaining change amount. Therefore, the account-based model essentially relies on address-reuse on the protocol level. This behavior practically renders the account-based cryptocurrencies inferior to UTXO-based currencies from a privacy perspective.

Previously, several works had identified the privacy shortcomings of the account-based model, specifically in Ethereum. Those works had proposed trustless coin mixers~\cite{meiklejohn2018mobius,seres2019mixeth,shlomovits2019sharelock} and confidential transactions~\cite{williamson2018aztec,bunz2019zether,cryptoeprint:2019:319}. However, until recently, none of these schemes has been deployed on Ethereum. Even today, Ethereum's privacy-enhancing overlays are still in a nascent, immature phase especially in comparison with Bitcoin's well-established coin mixer scene
~\cite{bonneau2014mixcoin,valenta2015blindcoin,ziegeldorf2015coinparty,heilman2017tumblebit,ruffing2014coinshuffle}.

\textbf{Our contributions:}
\begin{itemize}
    \item We identify and apply several quasi-identifiers stemming from address reuse (time-of-day activity, transaction fee, transaction graph), which allow us to profile and deanonymize Ethereum users.
    \item In the cryptocurrency domain, we are the first to quantitatively assess the performance of a recent area of machine learning in graphs, the so-called node embedding algorithms.
    \item We establish several heuristics to decrease the privacy guarantees of non-custodial mixers on Ethereum.
    \item We describe a version of the malicious value fingerprinting attack, also known as Danaan-gift attack~\cite{biryukov2019privacy}, applicable in Ethereum.
    \item We collect and analyze a wide source of Etherum related data, including Ethereum name service (ENS), Etherscan blockchain explorer, Tornado Cash mixer contracts, and Twitter. We release the collected data as well as our source code for further research\ifarxiv\footnote{\url{https://github.com/ferencberes/ethereum-privacy}}\fi.
\end{itemize}

The rest of the paper is organized as follows. In Section~\ref{sec:relatedwork}, we review related work. In Section~\ref{sec:background}, a brief background is given on the inner workings of Ethereum along with the general idea behind node embedding. In Section~\ref{sec:data}, we describe our collected data.
In Section~\ref{sec:evaluationmeasures}, we overview the literature on evaluating deanonymization methods and propose our metrics.
Our main methods to pair Ethereum addresses that belong to the same user and link Tornado deposits and withdrawals are detailed in Section~\ref{sec:pairing} and~\ref{sec:deanonymising}. A variant of the Danaan-gift attack is described in Section~\ref{sec:danaangift}. Finally, we conclude our paper in Section~\ref{sec:conclusion}. 
    
\section{Related Work} \label{sec:relatedwork}
First results on  Ethereum deanonymization~\cite{Klusman2018DeanonymisationIE} attempted to directly apply both on-chain and peer-to-peer (P2P) Bitcoin deanonymization techniques. The starting point of our work is the recognition that common deanonymization methods for Bitcoin \emph{are not} applicable to Ethereum due to differences in Ethereum's P2P stack and account-based model.

The relevant body of more recent literature takes two different approaches. The first analyzes Ethereum smart contracts with unsupervised clustering techniques~\cite{norvill2017automated}. Kiffer et al.~\cite{kiffer2018analyzing} assert a large degree of code reuse which might be problematic in case of vulnerable and buggy contracts.

The second branch of literature assesses  Ethereum addresses. A crude and initial analysis had been made by Payette et al., who clusters the Ethereum address space into only four different groups~\cite{payette2017characterizing}. More interestingly Friedhelm Victor proposes address clustering techniques based on participation in certain airdrops and ICOs~\cite{victoraddress}. These techniques are indeed powerful, however, they do not generalize well as it assumes participation in certain on-chain events. Our techniques are more general and are applicable to all Ethereum addresses. Victor et al. gave a comprehensive measurement study of Ethereum's ERC-20 token networks, which further facilitates the deanonymization of ERC-20 token holders~\cite{victor2019measuring}.

A completely different and unique approach is taken by~\cite{Shlomi2019}, which uses stylometry to deanonymize smart contract authors and their respective accounts. The work had been used to identify scams on Ethereum.

\section{Background} \label{sec:background}
In this section we provide some background on cryptocurrency privacy-enhancing technologies as well as node embedding algorithms. Elementary preliminaries on Ethereum and its applied gas mechanism are included in Appendix~\ref{sec:ethbasics}.

\subsection{Non-custodial mixers}
Coin mixing is a prevalent technique to enhance transaction privacy of cryptocurrency users. Coin mixers may be custodial or non-custodial. In case of custodial mixing, users send their ``tainted'' coins to a trusted party, who in return sends back ``clean'' coins after some timeout. This solution is not satisfactory as the user does not retain ownership of her coins during the course of mixing. Hence, the trusted mixing party might just steal funds, as it already happened with custodial mixers~\cite{meiklejohn2013fistful}. 
\begin{figure}
    \centering
    \includegraphics[width=0.5\textwidth,trim={4cm 9.5cm 2cm 10cm},clip]{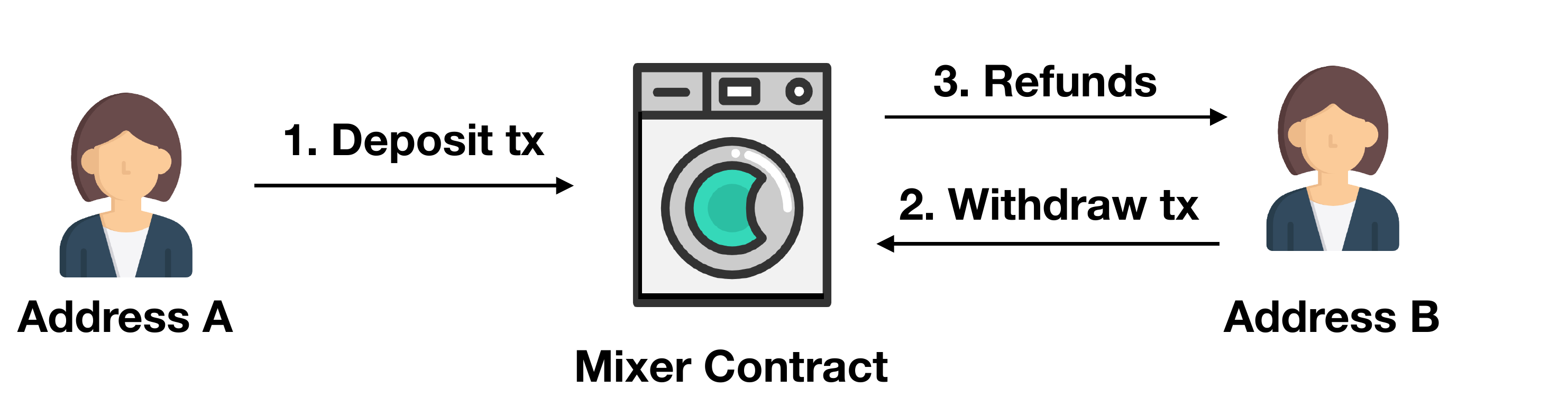}
    \caption{Schematic depiction of non-custodial mixers on Ethereum}
    \label{fig:mixerexplanation}
\end{figure}

Motivated by these drawbacks, recently several non-custodial mixers have been  proposed in the literature~\cite{meiklejohn2018mobius,whitehat2018miximus,seres2019mixeth,shlomovits2019sharelock}. The recurring theme of non-custodial mixers is to replace the trusted mixing party with a publicly verifiable transparent smart contract or with secure multi-party computation (MPC). Non-custodial mixing is a two-step procedure. First, users deposit equal amounts of ether or other tokens into a mixer contract from an address $\mathcal{A}$, see Figure~\ref{fig:mixerexplanation}. After some user-defined time interval, they can withdraw their deposited coins with a withdraw transaction to a fresh address $\mathcal{B}$. In the withdraw transaction, users can prove to the mixer contract that they deposited without revealing which deposit transaction was issued by them by using one of several available cryptographic techniques, including ring signatures~\cite{meiklejohn2018mobius}, verifiable shuffles~\cite{seres2019mixeth}, threshold signatures~\cite{shlomovits2019sharelock}, and zkSNARKs~\cite{whitehat2018miximus}.

\subsection{Ethereum Name Service}
Ethereum Name Service (ENS) is a distributed, open, and extensible naming system based on the Ethereum blockchain. In spirit it is similar to the well-known Domain Name Service (DNS). However, in ENS the registry is implemented in Ethereum smart contracts\footnote{See: \url{https://docs.ens.domains}}, hence it is resistant to DoS attacks and data tampering. Like DNS, ENS operates on a system of dot-separated hierarchical names called domains, with the owner of a domain having full control over subdomains. ENS maps human-readable names like \texttt{alice.eth} to machine-readable identifiers such as Ethereum addresses. Therefore, ENS provides a more user-friendly way of transferring assets on Ethereum, where users can use ENS names (\texttt{alice.eth}) as recipient addresses instead of the error-prone hexadecimal Ethereum addresses.  

\subsection{Node embeddings}\label{sec:background_embeddings}

Node embedding methods form a class of network representation learning methods that map graph nodes to vectors in a low-dimensional vector space. They are designed to represent vertices with similar neighborhood structure by vectors that are close in the vector space. Intuitively, addresses that interact with the same set of addresses in the Ethereum transaction graph should be close in the embedded space. Perhaps the best methods are Laplacian eigenmaps~\cite{laplacian} and graph factorization~\cite{graphfactorization}. Research in node embedding has  recently been catalyzed by Word2Vec~\cite{word2vec}, an embedding method for natural language processing. Several node embedding methods have been proposed recently \cite{deepwalk,line,node2vec,netmf} and applied successfully for multi-label classification and link prediction in a variety of real-world networks from diverse domains. In this work, we use these techniques on the Ethereum transaction graph to link addresses owned by the same user. To the best of our knowledge, we are the first to apply node embedding for Ethereum user profiling.

\begin{figure}[t]
    \includegraphics[width=0.5\textwidth]{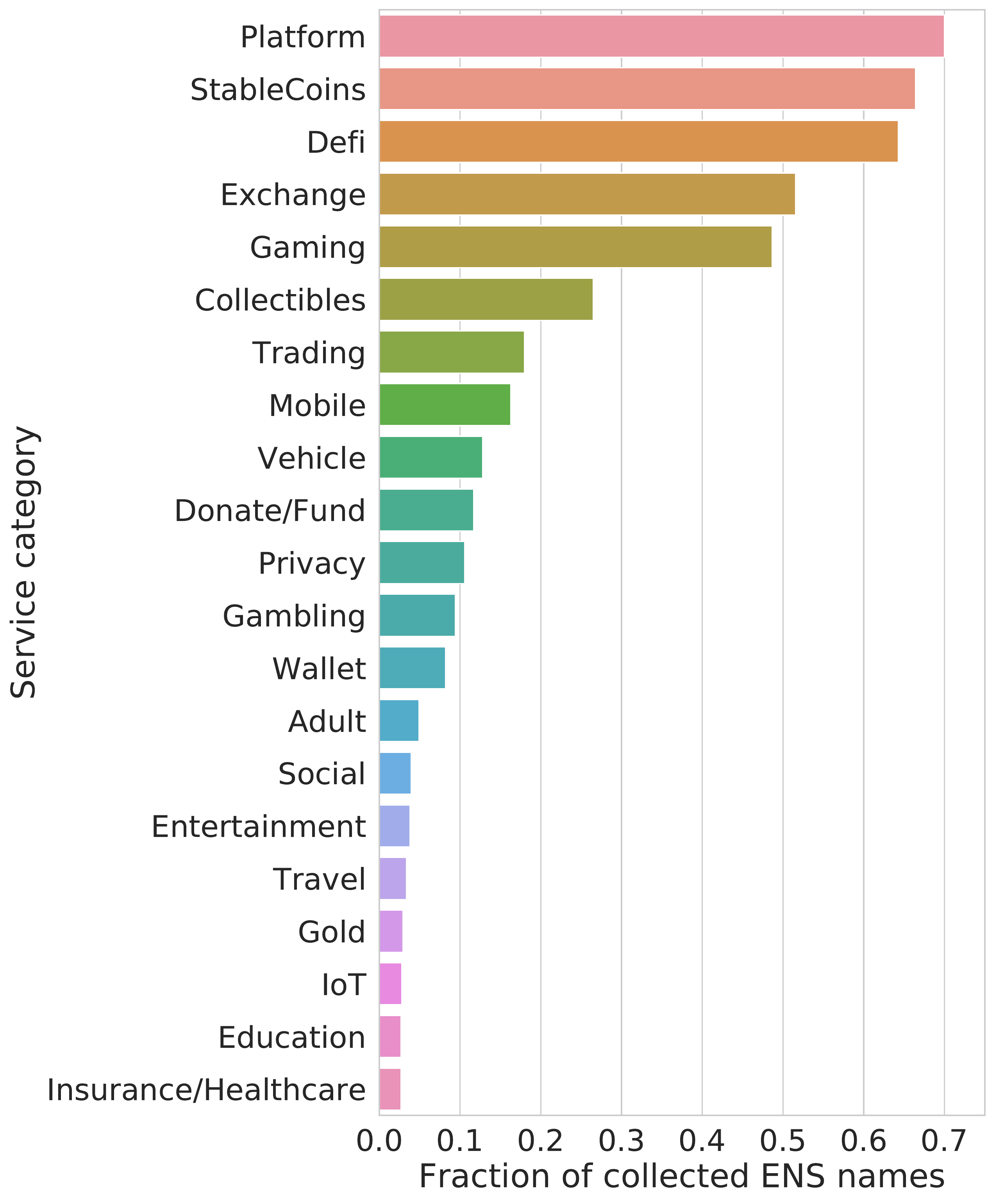}
    \caption{Fraction of ENS names (collected from Twitter) that interacted with the given service topics. Popular services within the categories are shown in Figure~\ref{fig:service_names}.}
    \label{fig:qualitativeaddressanalysis}
\end{figure}

\begin{figure}[t]
\centering
    \includegraphics[width=0.3\textwidth]{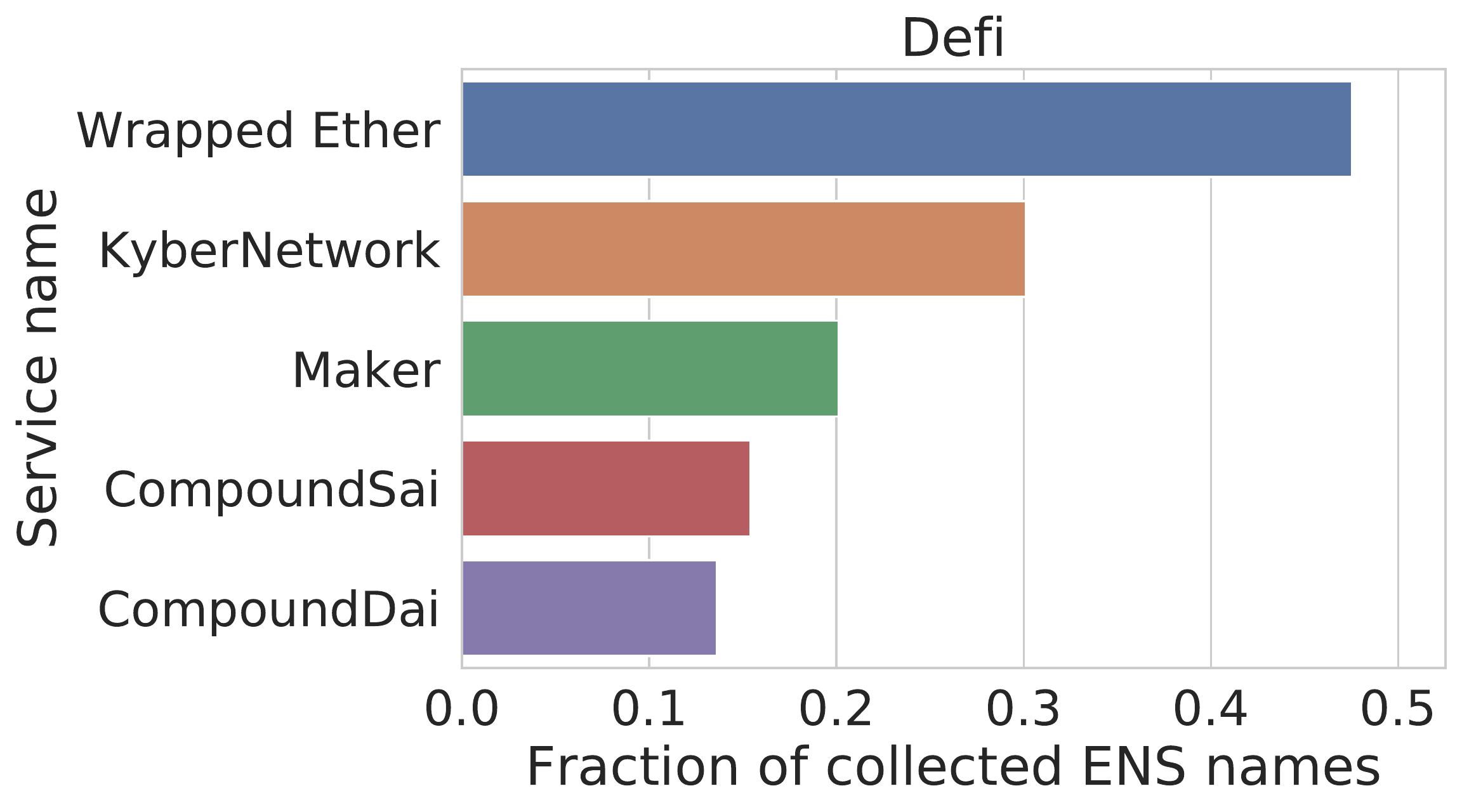}
    \includegraphics[width=0.3\textwidth]{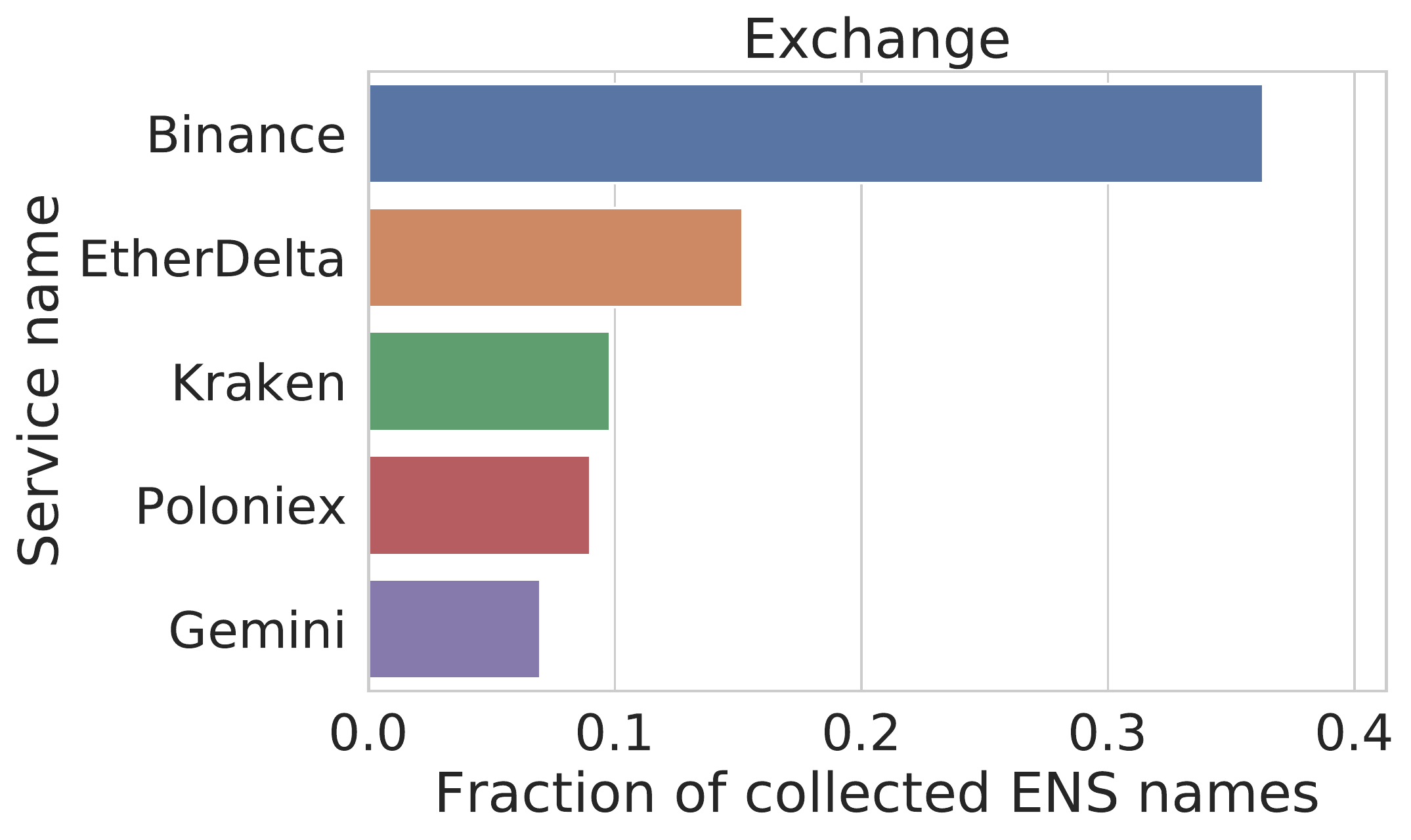}
    \includegraphics[width=0.3\textwidth]{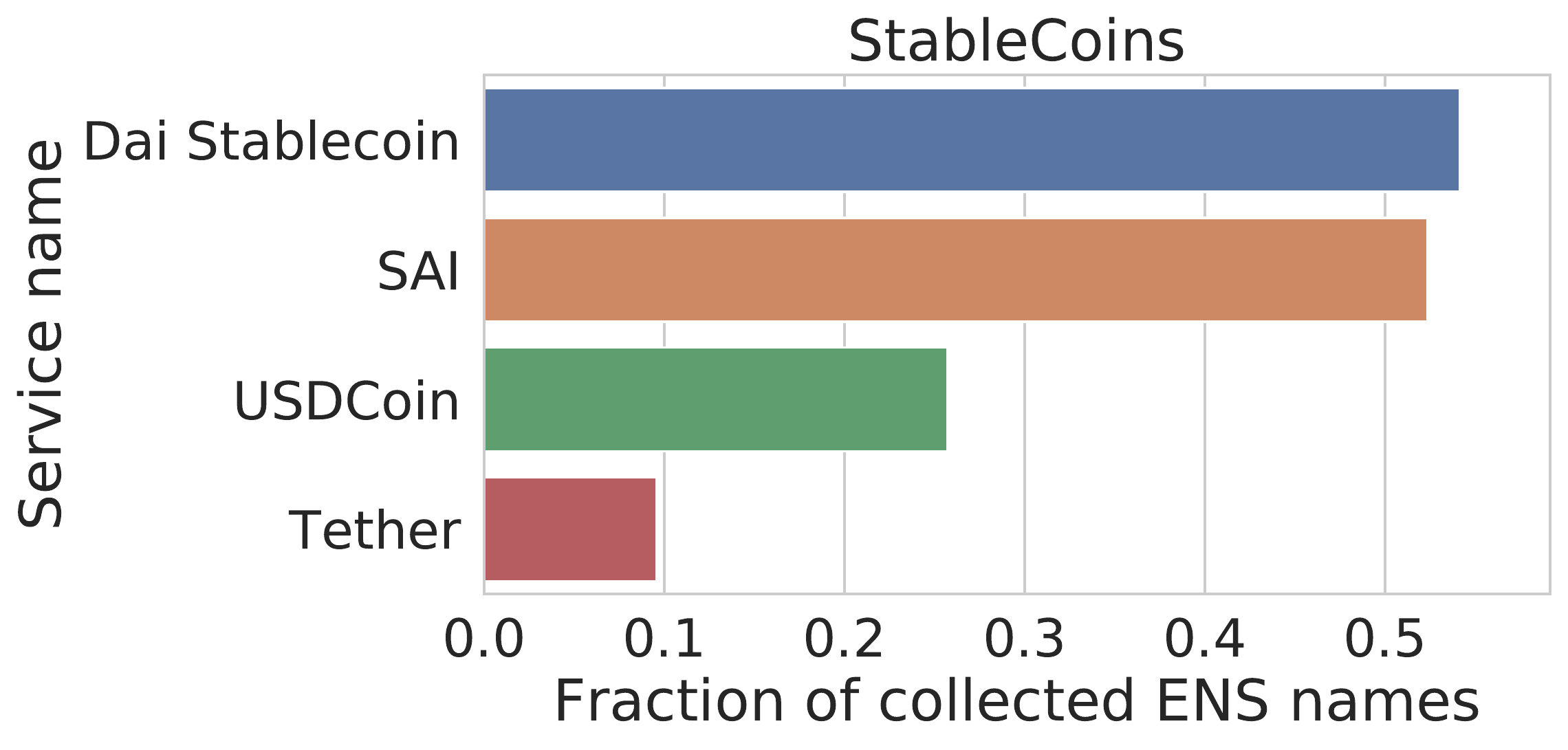}
    \includegraphics[width=0.3\textwidth]{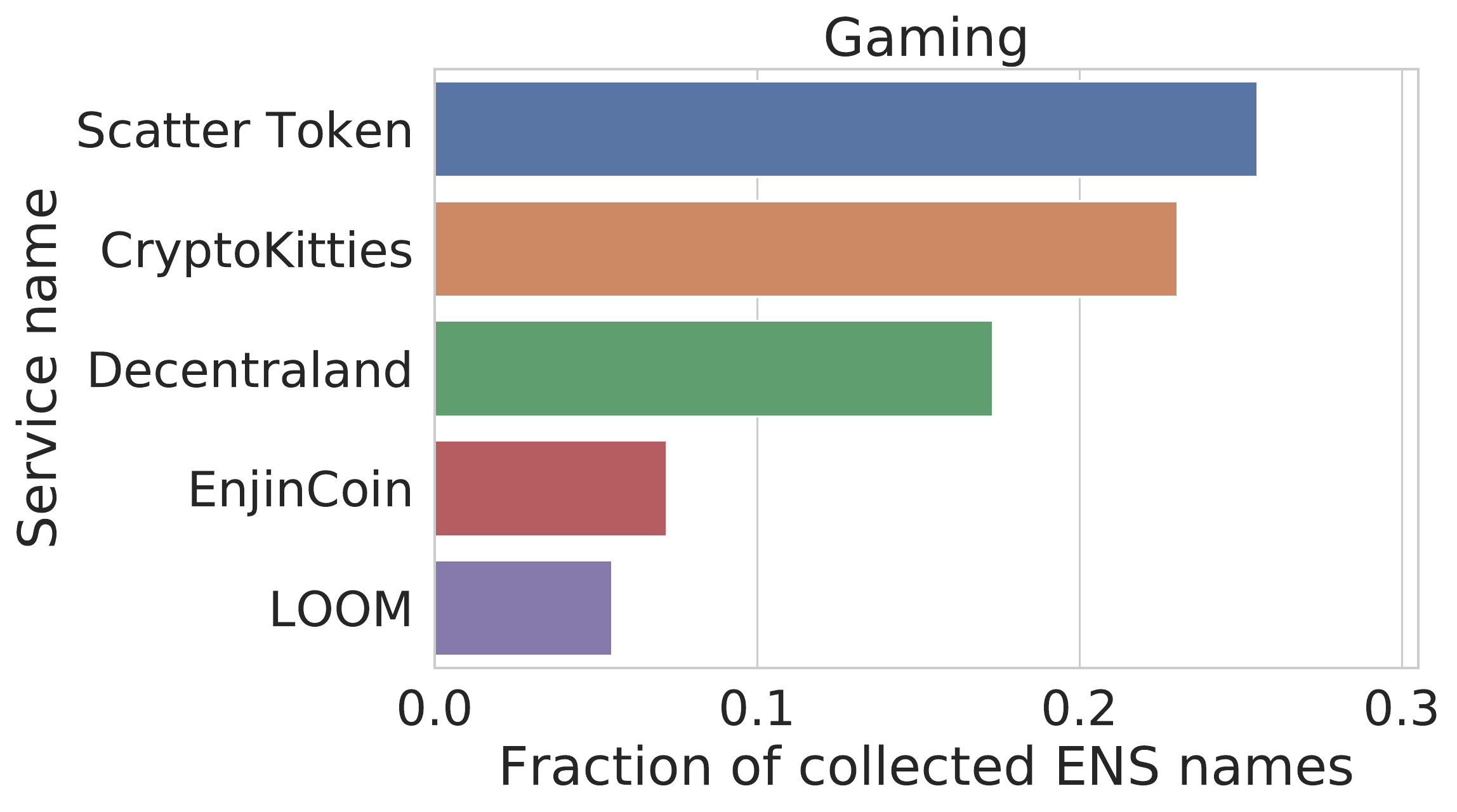}
    \caption{Most popular services within the Defi, Exchange, Stablecoins and Gaming categories in our data collection.}
    \label{fig:service_names}
\end{figure}

\section{Data collection} \label{sec:data}

We collected addresses presumably related to regular users and not automatic (trader or exchange) bots from the following publicly available data sources. \textbf{Twitter:} 
By using the Twitter API\footnote{Using the Twitter Search and People API endpoints, we collected tweets containing the following keywords \{'@ensdomains','.eth','ENS name','ENS address', 'ethereum', '\#ethereum'\} as well as profiles with an ENS name in their displayed profile name or description. We also searched for ENS names in the name and description of every tweeter in our data. Twitter data collection lasted from 2019-11-15 until 2020-03-05.}, we were able to collect 890 ENS names included in Twitter profile names or descriptions, and discover the connected Ethereum addresses, see Figure~\ref{fig:addr_cnt_per_ens}. 
\textbf{Humanity DAO:}\footnote{See: \url{https://www.humanitydao.org/humans}}A human registry of Ethereum users, which can include a Twitter handle in addition to the Ethereum address. \textbf{Tornado Cash mixer contracts:} We collected all Ethereum addresses that issued or received transactions from Tornado Cash mixers up to 2020-04-04. 
Table~\ref{tab:address_sources} shows the total number of addresses collected from each data source as well as addresses with at least $5$ sent transactions. We note that there are overlaps between the three address groups, see the last row of Table~\ref{tab:address_sources}.

By using the Etherscan blockchain explorer API, we collected 1,155,188 transactions sent or received by the addresses in our collection. The final transaction graph contains 159,339 addresses. The transactions span from 2015-07-30 until 2020-04-04.
Figure~\ref{fig:avg_txs} shows the average number of transactions sent and received in the three data sources. Addresses collected from Twitter and Humanity DAO have similar behavior, while Tornado accounts have fewer transactions since Tornado Cash has only recently been launched.

Finally, using the Etherscan Label Word Cloud, we manually collected service category labels (e.g.\ exchange, gambling, stablecoins) related to popular addresses in our data set. We summarize the fraction of ENS names in our collection that interacted with the given services in Figure~\ref{fig:qualitativeaddressanalysis}. We observed that the publicly revealed ENS names already expose sensitive activities such as gambling and adult services. Therefore, users should avoid sensitive activities on addresses easily linkable to their public identities, such as ENS name or their Twitter handle.

\begin{table}
    \centering
    \begin{tabular}{l|r|r|r}
         \textbf{Source} & \textbf{Total} & \textbf{At least $5$} & \textbf{Used as ground} \\
         &  & \textbf{sent txs} & \textbf{truth pairs}   \\
         \hline
         Twitter & 1364 & 1260 & 129 \\
         Tornado Cash & 2361 & 1618 & $^*$189 \\
         Humanity-Dao & 695 & 602 & n/a \\
         \hline
         All & 4259 & 3321 & 318 \\
    \end{tabular}
    \caption{Number of Ethereum addresses collected from three different sources. $^*$Tornado ground truth pairs are only heuristically identified, see Section~\ref{sec:heuristics}. Due to overlaps between the data sources, the total number of investigated addresses is less than the sum of the records in the top three rows.}
    \label{tab:address_sources}
\end{table}

\begin{figure}
    \centering
    \includegraphics[width=0.3\textwidth]{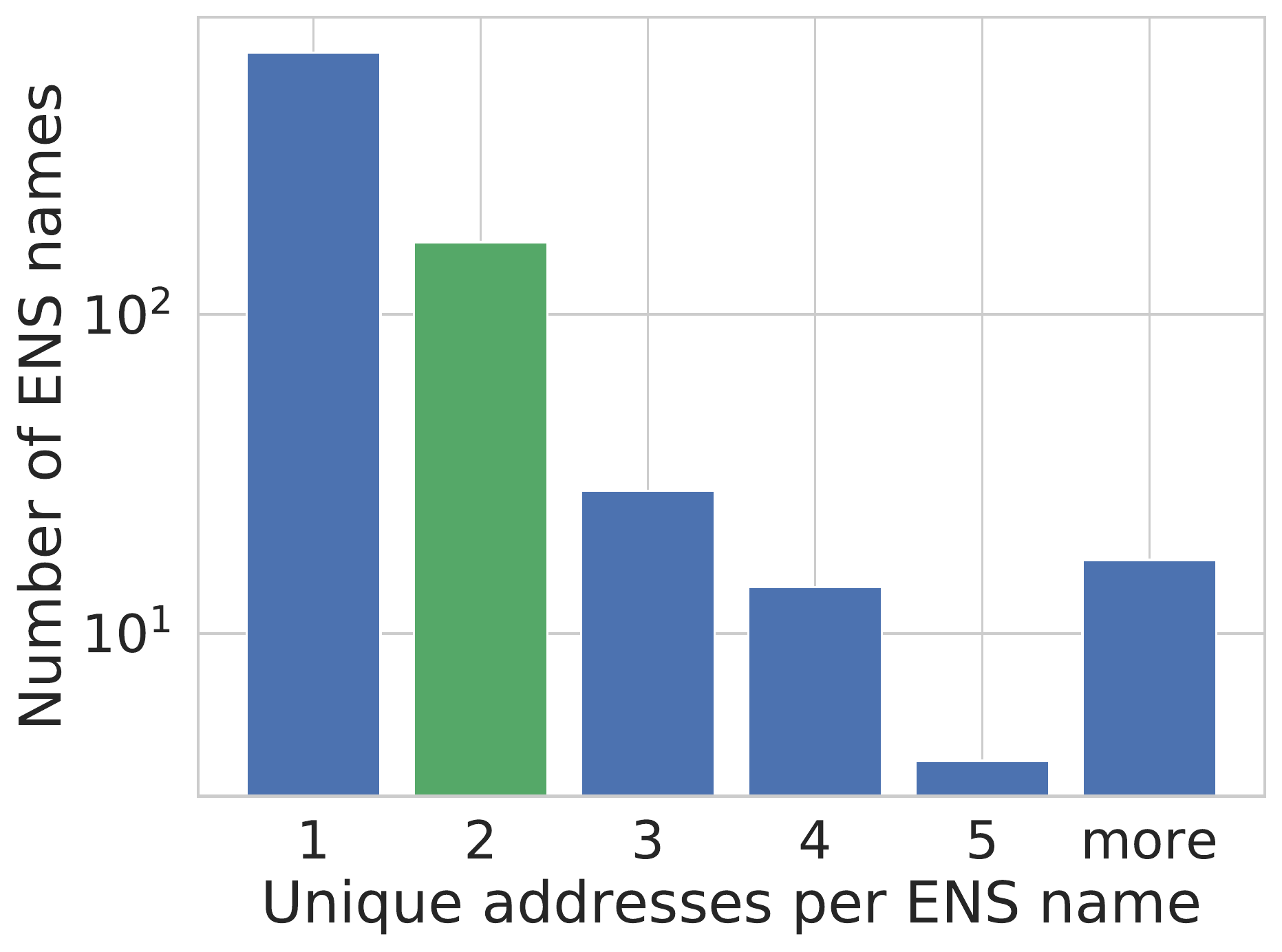}
    \caption{Unique address count of ENS names collected from Twitter. Most of the ENS names in our collection are linked to a single Ethereum address, while some entities use multiple accounts. In Section~\ref{sec:pairing}, we use ENS names with exactly two unique addresses \textbf{(green)} to measure the performance of different profiling techniques.}
    \label{fig:addr_cnt_per_ens}
\end{figure}

\begin{figure}
    \centering
    \includegraphics[width=0.3\textwidth]{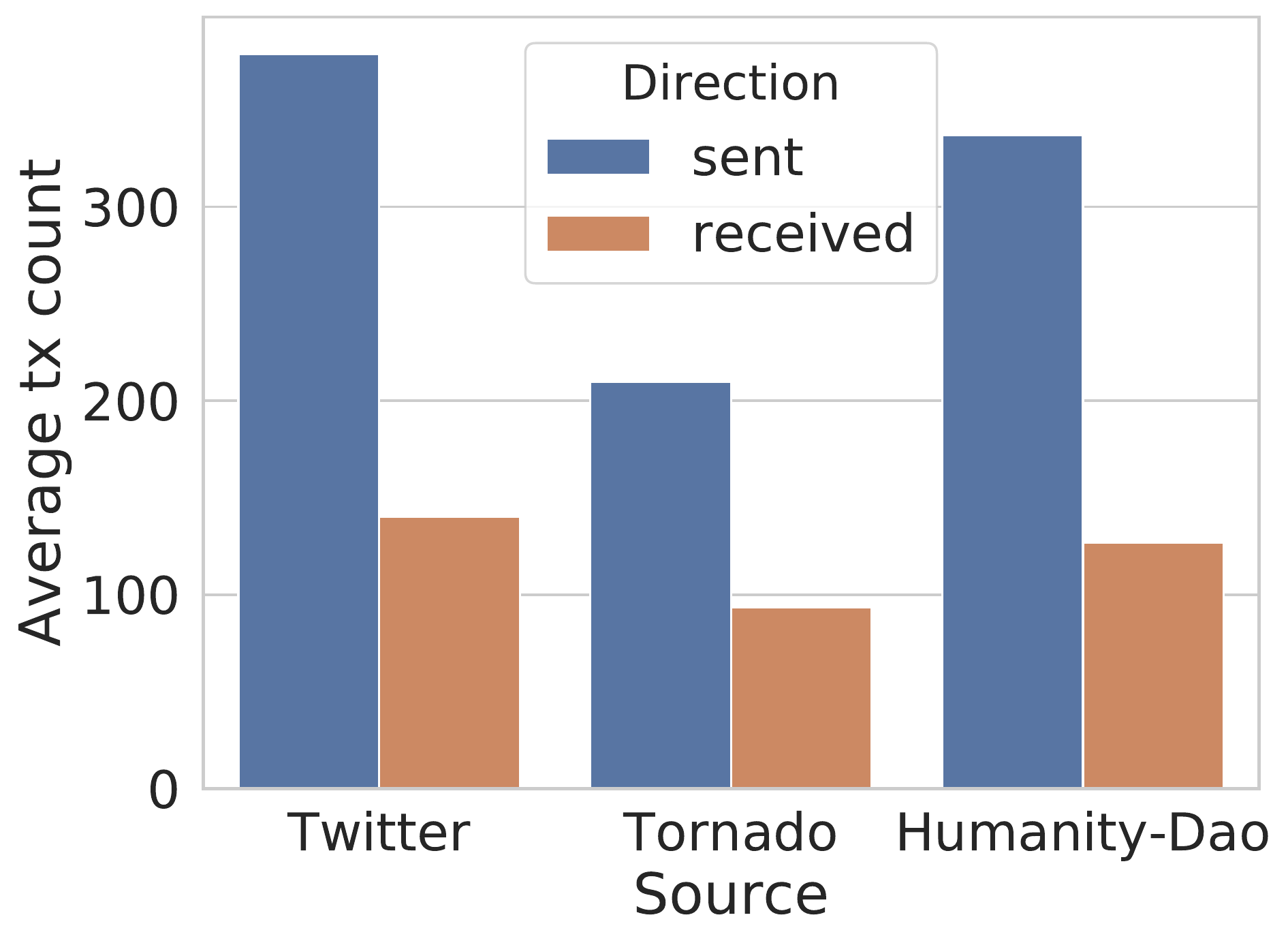}
    \caption{Average number of transactions sent or received by the  addresses of each data source. Tornado accounts have less transactions as the service has only recently been launched.}
    \label{fig:avg_txs}
\end{figure}

\section{Evaluation measures}\label{sec:evaluationmeasures}

In this paper, we propose deanonymization methods for pairing Etherum accounts of the same user (Section~\ref{sec:pairing}), Tornado deposits and withdrawals (Section~\ref{sec:deanonymising}), and fingerprinting accounts (Section~\ref{sec:danaangift}).
To establish an appropriate measure for evaluating our methods, we face the diversity and complexity of estimates of the adversary’s success to
breach privacy. In the literature, the adversary's output takes the form of a posterior probability distribution, see the survey~\cite{wagner2018technical}.

The simplest metrics consider the success rate of a deanonymizing adversary. Metrics such as accuracy, coverage, fraction of correctly identified nodes~\cite{backstrom2007wherefore,narayanan2009anonymizing,narayanan2011link} are applicable only when the attack has the potential to exactly identify a significant part of the network.

Exact identification is an overly ambitious goal in our experiments, which aim to use very limited public information to rank candidate pairs and quantify the leaked information as risk for a potential systematic deanonymization attack.
For this reason, we quantify non-exact matches, since even though our deanonymizing tools might not exactly find a mixing address, they can radically reduce the anonymity set, which is still harmful to privacy.
We want to quantify the information leaked from network structure,  time-of-day activity, and gas price usage to assess the implications for the \emph{future privacy}~\cite{narayanan2008robust} of the account owners.

In our first two deanonymization experiments, our algorithms will return a ranked list of candidate pairs for each account in our testing set.
Based on the ranked list, we
propose a simple metric, the \textbf{average rank} of the target in the output.

Recent results consider deanonymization as a classification task and use AUC for evaluation~\cite{ma2017anonymizing}. In our experiments, we will compute AUC by the following claim:

\begin{lemma}
Consider a set of accounts $a$, each with a set of candidate pairs $c(a)$ such that exactly one in $c(a)$ is the correct pair of $a$. Let an algorithm return a ranked list of all sets $c(a)$.  The AUC of this algorithm is equal to the average of $r(a)/|c(a)|$ over all $a$, where $r(a)$ is the rank of the correct pair of $a$ in the output.
\end{lemma}

\begin{proof}
  Follows since AUC is the probability that a randomly selected correct record pair is ranked higher than another incorrect one~\cite{hanley1982meaning}.
\end{proof}
  

Finally, we consider evaluation by variants of entropy, which quantify privacy loss by the number of bits of additional information needed to identify a node.
Defining entropy is difficult in our case for two reasons. First, our algorithms provide a ranked list and not a probability distribution. Second, for the Tornado Cash mixer deanonymization, the anonymity set size is dynamic, as users can freely deposit anytime they wish, hence increasing the anonymity set size.

In the literature, entropy based evaluation considers the a priori knowledge without a deanonymization method and the a posteriori knowledge after applying one~\cite{serjantov2002towards}. Several papers compute the entropy of the a posteriori knowledge~\cite{serjantov2002towards,diaz2002towards,narayanan2008robust}, however they assume that the deanonymizer outputs a probability distribution of the candidate records~\cite{narayanan2008robust}. 


The information the attacker has learned with the attack can be expressed
as the difference of the a priori and a posteriori entropy. We call this difference the \textbf{entropy gain}, denoted as gain$(n, p)$ where $n$ and $p$ are the anonymity set size and probability distribution, respectively.
The a priori entropy of
the target record is typically the base-2 logarithm of the a priori anonymity set size.  
The problem with varying a priori anonymity set size is that while correctly selecting ten candidate users from a pool of a million is a great achievement, the same entropy of $\log_2(10)$ is achieved without deanonymization if the initial pool size, for example in a low-utilization mixer, is only 10.
We note that in~\cite{diaz2002towards}, the authors also divide the entropy gain to normalize the value.


Next, we describe a new method to infer the a posteriori distribution given varying a priori knowledge and appropriately normalize with respect to the a priori entropy.  More precisely, first we give a heuristic argument that the a priori anonymity set size has little effect on the entropy gain, and hence we can compare and average across different measurements. In the formula below, given an a priori anonymity set size $2n$ vs.~$n$, we compare the entropy gain of the same distribution $p$, gain$(2n,p)-{}$gain$(n,p)$. In the formula below, $p_i$ denotes the probability $p([(i-1)/(2n),i/(2n)])$.
\begin{eqnarray} 
\nonumber
    \text{gain}(2n,p) &=& \log_2(2n)+ \sum_{i=1}^{2n} p_i \log_2(p_i);\\
\nonumber
    \text{gain}(n,p) &=& \log_2(n)+ \sum_{i=1}^{n} (p_{2i-1} + p_{2i}) \log_2(p_{2i-1} + p_{2i}).
    \label{eq:entropy}
\end{eqnarray} 
Since $\log_2(2n)-\log_2(n) = 1 = \sum_i p_i$, we may group the terms to obtain the difference in the entropy gain as the sum for $1\le i \le n$ of
\begin{equation}
p_{2i-1} \log_2\left(\frac{2p_{2i-1}}{p_{2i-1}+ p_{2i}}\right) + p_{2i} \log_2\left(\frac{2p_{2i}}{p_{2i-1}+ p_{2i}}\right),
\end{equation}
which can be bounded from above by using $\log x <x-1$ as
\begin{equation}
    \frac{(p_{2i-1} - p_{2i})^2}{p_{2i-1} + p_{2i}}.
\end{equation}
If the probability distribution is smooth with little density changes in a neighborhood, the above value is very small. For example, the value is small if $p_i$ is monotonic in $i$, which at least approximately holds in our experiments.

Based on the above argument, we may infer an  \textbf{empirical probability distribution} of the candidates ranked by an algorithm. For each a priori size $n$ and rank $r$ for the ground truth pair of a target record, we define the distribution $P(n,r)$ to be uniform in $[(r-1)/n, r/n]$, and 0 elsewhere, in accordance with formula~(\ref{eq:entropy}). The empirical probability distribution of an algorithm will be the average of $P(n,r)$ over all the output of the algorithm.  In the discussion, we will use the \textbf{entropy gain} of the above empirical probability distribution to quantify the deanonymization power of our algorithms.




\section{Linking Ethereum accounts of the same user}\label{sec:pairing}

In this section, we introduce our approach to identify pairs of Ethereum accounts that belong to the same user. In our measurements, we investigate three quasi-identifiers of the account owner: the active time of the day, the gas price selection, and the location in the Ethereum transaction graph.

We evaluate our methods by using the set of address pairs in our collection that belong to the same name in the Ethereum Name Service (ENS), see  Figure~\ref{fig:addr_cnt_per_ens}. We consider 129 ENS names with exactly two Ethereum addresses to avoid the possible validation bias caused by ENS names with more than two addresses. We also note that Ethereum addresses connected to multiple ENS names were excluded from our experiments.

\subsection{Time-of-day transaction activity}

\begin{figure}[t]
    \centering
    \includegraphics[width=0.3\textwidth]{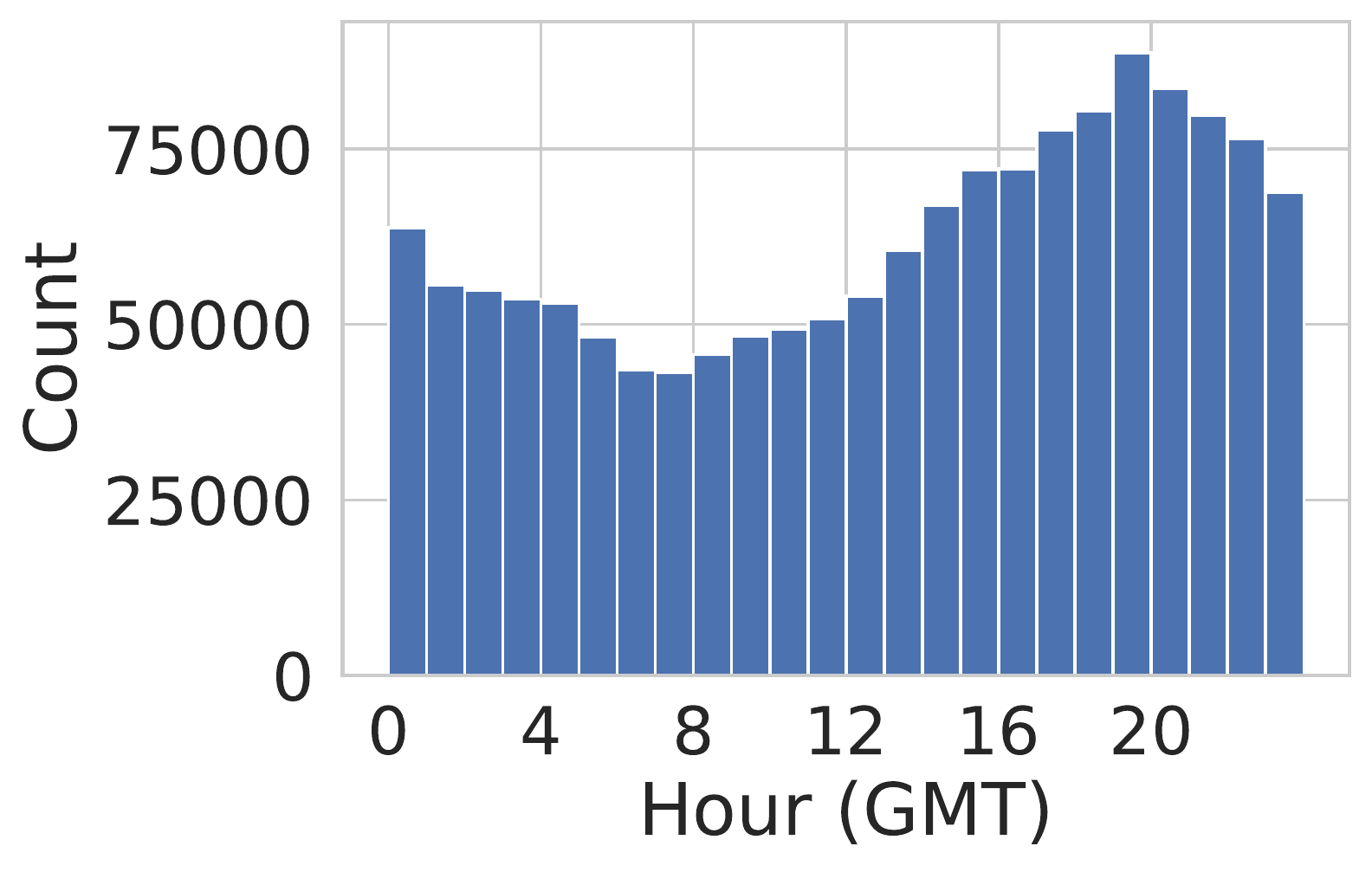}
    \caption{Time-of-day distribution of Ethereum transactions}
    \label{fig:acitivityHist}
\end{figure}

\begin{figure}[t]
    \centering
    \includegraphics[width=0.35\textwidth]{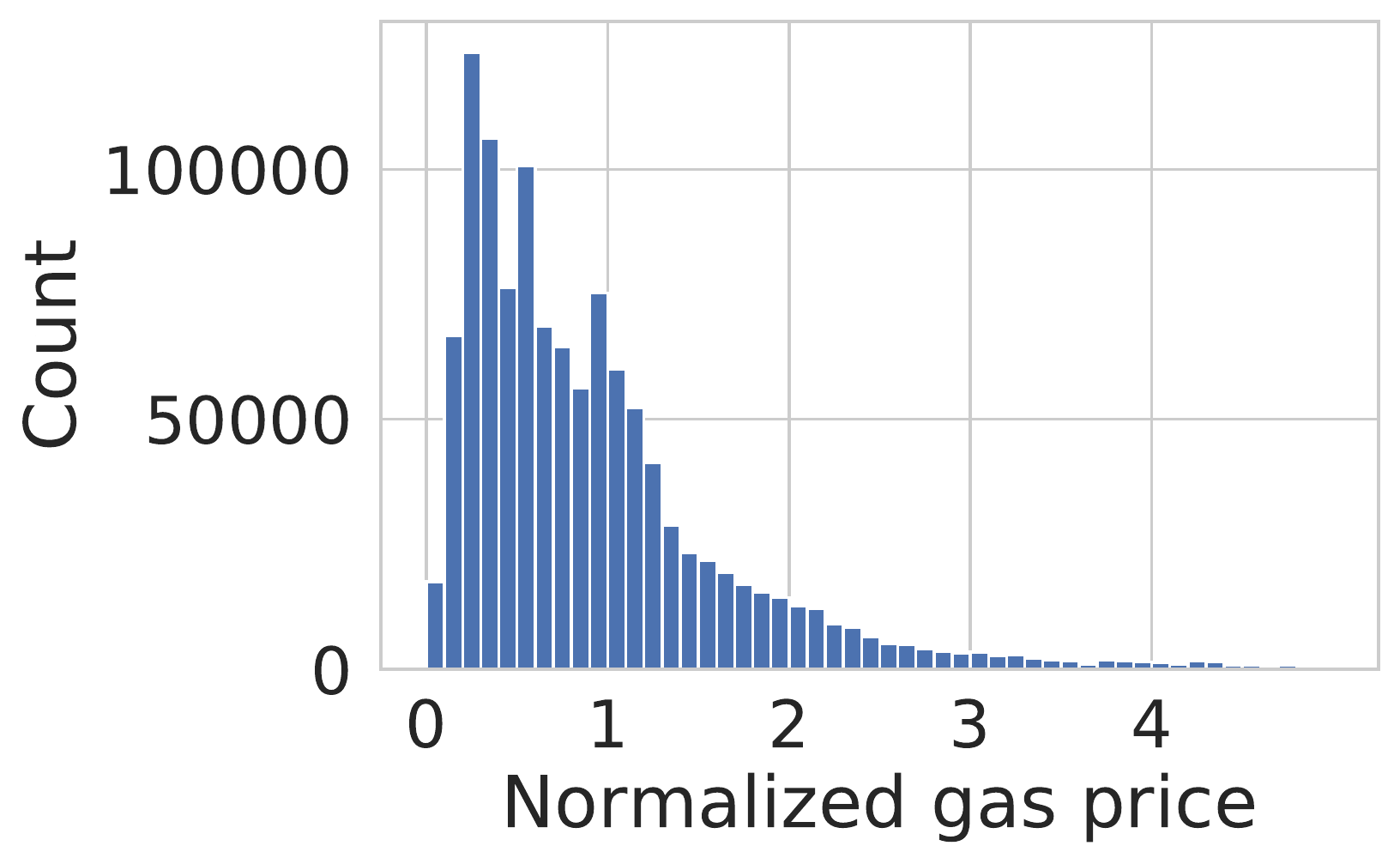}
    \caption{Normalized gas price distribution of Ethereum transactions. Outliers above $5$ are omitted.}
    \label{fig:normalizedGasPrice}
\end{figure}

Ethereum transaction timestamps reveal the daily activity patterns of the account owner, see Figure~\ref{fig:acitivityHist}. 
In the top row of Figure~\ref{fig:ethprofile}, we show time-of-day profiles for two ENS names that are active in different time zones.

Given the set of timestamps, an account is represented by the vector including the mean, median and standard deviation, as well as the time-of-day activity histogram divided into $b_{\text{hour}}$ bins.

\subsection{Gas price distribution}
Ethereum transactions also contain the gas price, which is usually automatically set by wallet softwares. Users rarely change this setting  manually. Most wallet user interfaces offer three levels of gas prices, slow, average, and fast where the fast gas price guarantees almost immediate inclusion in the blockchain.

The changes in daily Ethereum traffic volume sometimes cause temporary network  congestion, which affect user gas prices. Hence we normalized the gas price by the daily network average. In Figure~\ref{fig:normalizedGasPrice}, the two peaks of the normalized gas price around $0.5$ and $1$ correspond to the slow and average gas price options. On the other hand, users only occasionally charge more than three times the daily average gas price. The combination of these gas price levels forms the so-called gas price profile for each Ethereum user.

Given the normalized gas prices of the transactions sent, an account is represented by the vector including the mean, median and standard deviation, as well as the normalized gas price histogram divided into $b_{\text{gas}}$ bins.

\begin{figure}
    \includegraphics[width=0.45\textwidth]{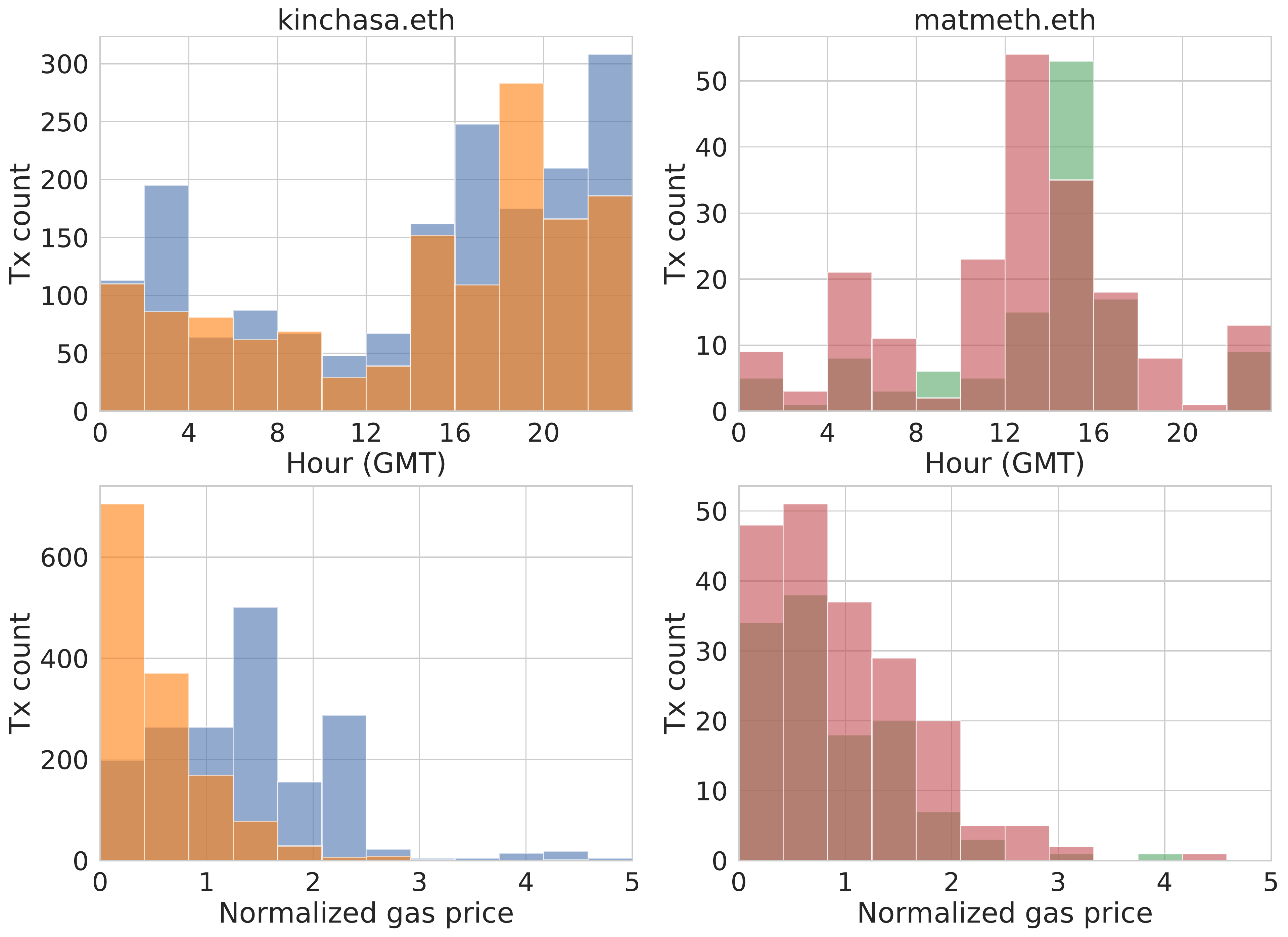}
    \caption{Time-of-day and normalized gas price profiles for two ENS names with a pair of addresses each. Both the time-of-day and gas price selection are similar in case of matmeth.eth addresses (red, green) while the addresses of kinchase.eth (blue, orange) have different gas price profiles. Addresses are denoted by different colors.}
    \label{fig:ethprofile}
\end{figure}

\subsection{Transaction graph analysis}
\label{sec:node_embeddings}
The set of addresses used in interactions characterize a user.  Users with multiple accounts might interact with the same addresses or services from most of them. Furthermore, as users move funds between their personal addresses, they may unintentionally reveal their address clusters. 

Our deanonymization experiments are conducted on a transaction graph with nodes as Ethereum addresses and edges as transactions.  From the library\footnote{https://github.com/benedekrozemberczki/karateclub} of Rozenberczki et al.~\cite{karateclub2020}, we selected twelve node embedding methods~\cite{boostne, nodesketch, netmf, walklets, hope, grarep, deepwalk, nmfadmm, laplacian,graphwave, diff2vec} (see Section~\ref{sec:background_embeddings}) to discover address pairs that might belong to the same user. To the best of our knowledge, we are the first to apply node embedding for Ethereum user profiling.

To apply the selected library~\cite{karateclub2020}, certain preprocessing steps are required. First, we considered transactions as undirected edges and removed loops and multi-edges. We excluded nodes outside the largest connected component. Due to running time considerations, we also removed nodes with degree one. The resulting graph has 16,704 nodes and 132,231 edges. We generated $128$-dimensional representations for the addresses. In order to compare with timestamp and gas price representations, we assign the overall average of the network embedding vectors to the removed nodes.

\subsection{Evaluation}

Based on timestamp, gas price distributions or  network embedding, we generate Euclidean feature vectors for $3321$ Ethereum addresses with each having at least five transactions sent, see Table~\ref{tab:address_sources}.
Given a target address, we order the remaining addresses by their Euclidean distance from the target.

In the evaluation, we use 129 address pairs that belong to the same ENS name. The accuracy metrics of Section~\ref{sec:evaluationmeasures} for identifying accounts of the same user by using only time-of-day activity or  normalized gas price is given in Figures~\ref{fig:granularity-avgrank}--\ref{fig:granularity-entropy}.
While time-of-day representation works best with $b_{\text{hour}}=4$ to 6 (six to four hour long bins), normalized gas price representation performs weaker and the related histogram gives only very small improvement with $b_{\text{gas}}=50$ over mean, median and standard deviation.

The performance of the twelve different node embedding algorithms is shown in Figures~\ref{fig:node_embeddings-avgrank}--\ref{fig:node_embeddings-entropy} based on ten independent experiments.
The two best performing  methods are Diff2Vec~\cite{diff2vec} and Role2Vec~\cite{role2vec}. Note that these algorithms capture different aspects of the same graph as Diff2Vec is a neighbourhood preserving and Role2Vec is a structural node embedding. We achieved best Ethereum address linking performance by combining these two methods by the harmonic average of their rank.

In Figure~\ref{fig:ens_performance}, we show the fraction of pairs where the rank of the ground truth pair is not more than a given value. Surprisingly, Diff2Vec and Role2Vec find the corresponding ENS address pairs within $100$ closest representations by almost $20\%$ more likely than time-of-day activity and gas price statistics. Our combination based approach  further improves the performance.

Our results show that the proposed profiling techniques link Ethereum addresses of the same user  significantly better than random guessing. More precisely, the combination of Diff2Vec and Role2Vec yield $1.6$ bits of additional information on account owners, see Figure~\ref{fig:node_embeddings-entropy}. In other words, we can reduce the anonymity set of a particular address by a factor of $2^{1.6}\approx 3.0314$.

\begin{figure}
    \centering
    \includegraphics[width=0.4\textwidth]{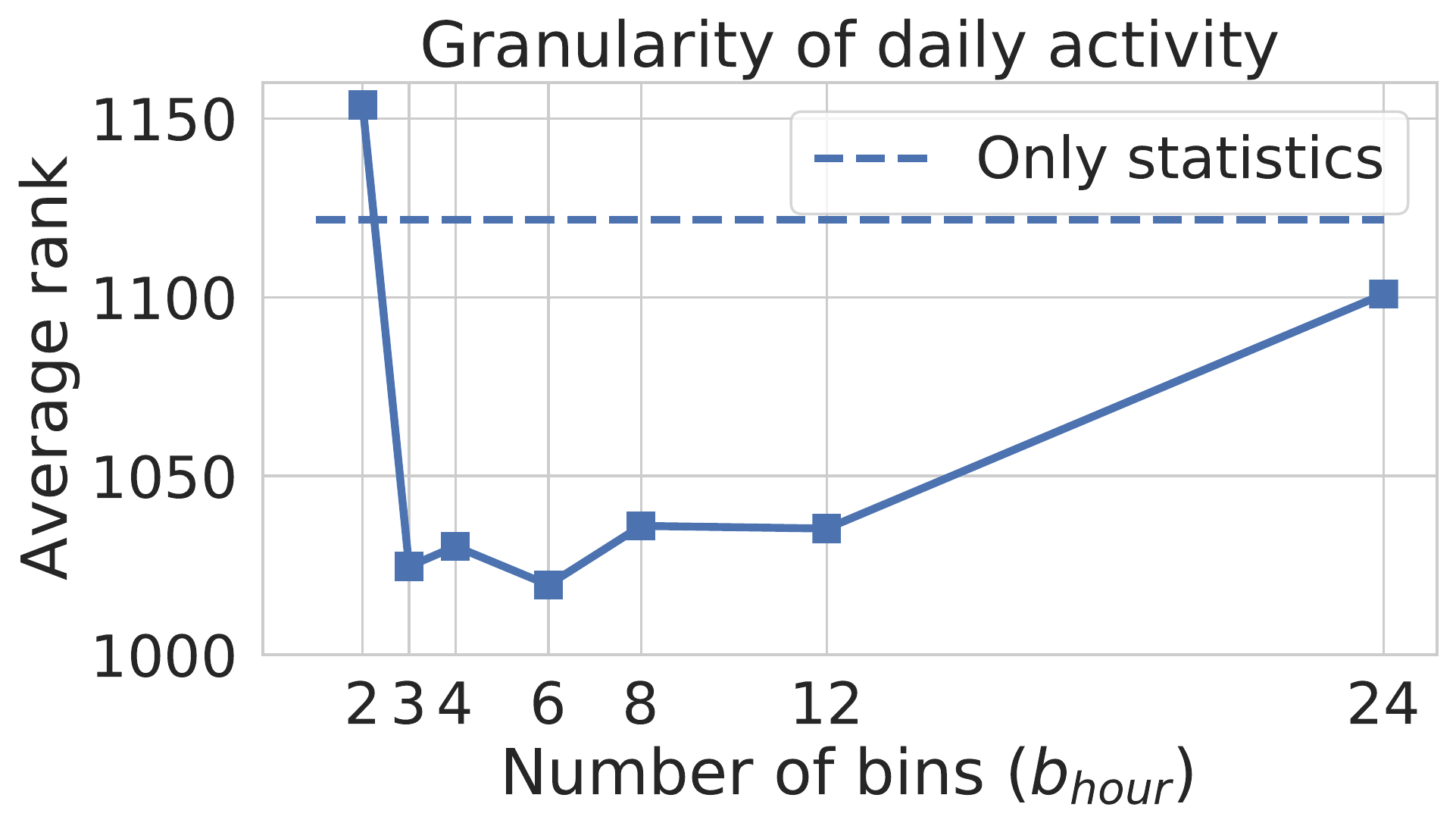}
    \includegraphics[width=0.4\textwidth]{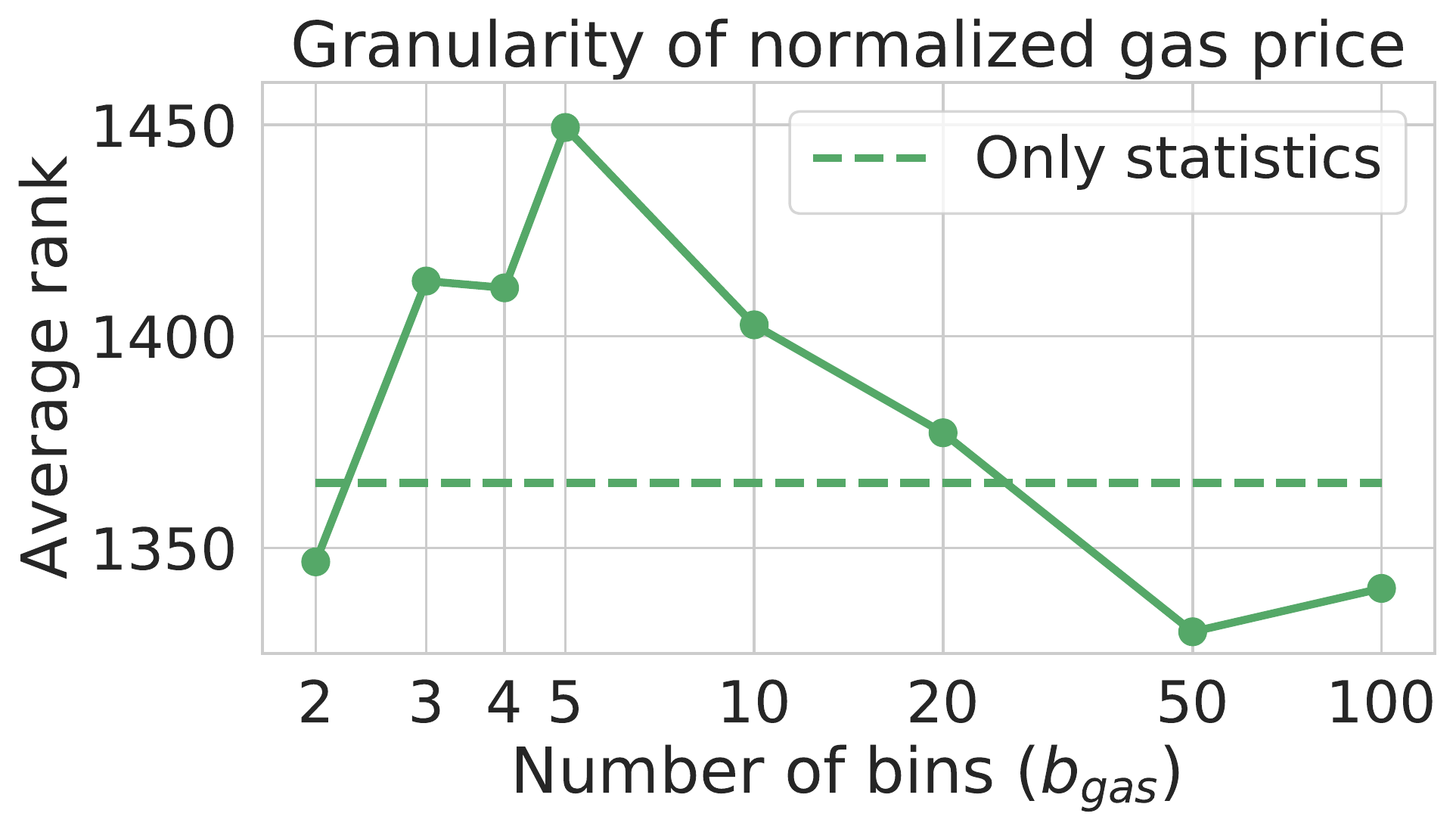}
    \caption{Average rank at different granularity for daily activity \textbf{(top)} and normalized gas price \textbf{(bottom)}. \textbf{Dashed lines} show performance with only mean, median and standard deviation used. Note that the maximum rank is $3321$, the total number of Ethereum addresses considered in this experiment.}
    \label{fig:granularity-avgrank}
\end{figure}

\begin{figure}
    \centering
    \includegraphics[width=0.4\textwidth]{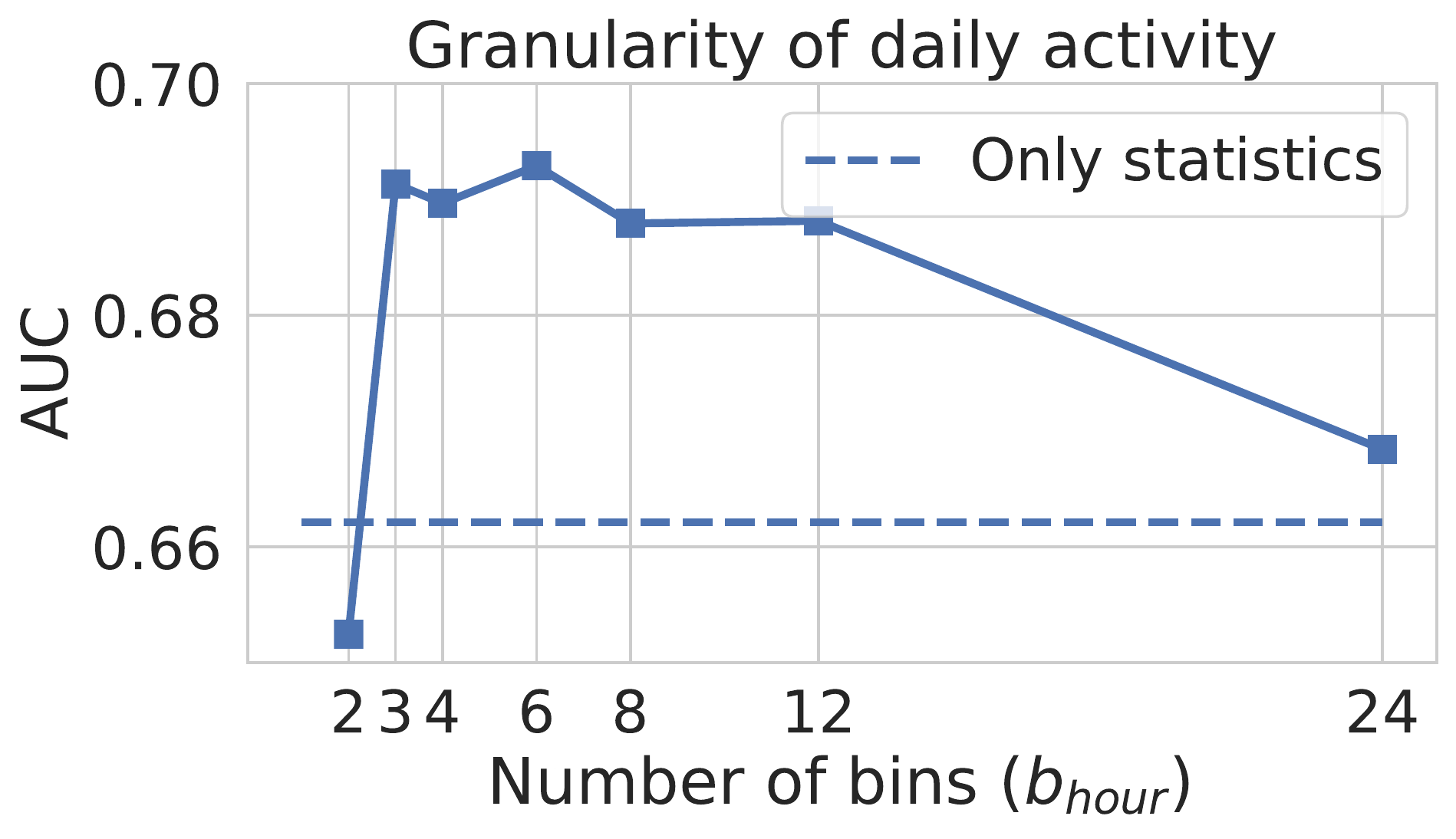}
    \includegraphics[width=0.4\textwidth]{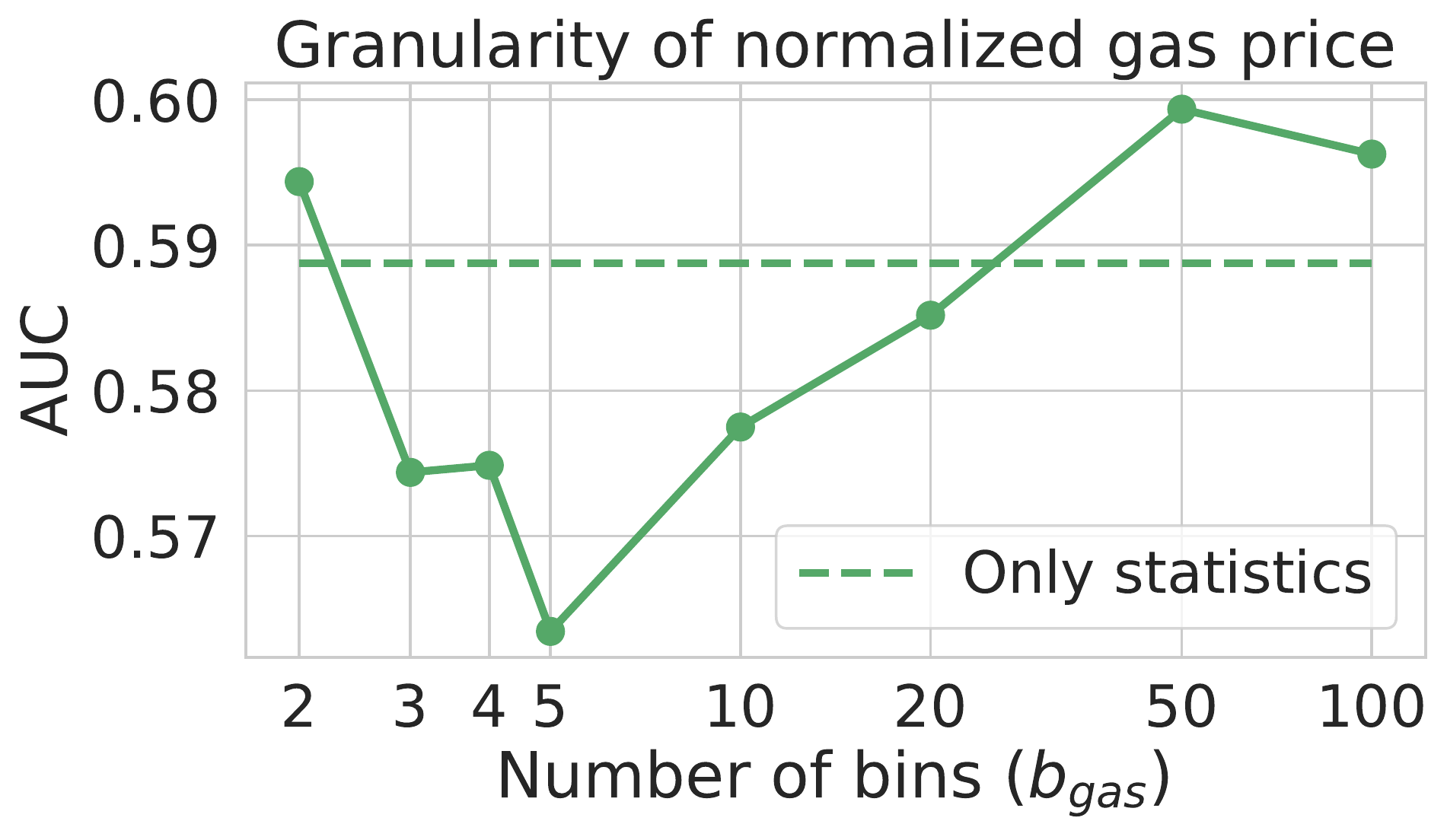}
    \caption{AUC at different granularity for daily activity \textbf{(top)} and normalized gas price \textbf{(bottom)}. \textbf{Dashed lines} show performance with only mean, median and standard deviation used.}
    \label{fig:granularity-auc}
\end{figure}

\begin{figure}
    \centering
    \includegraphics[width=0.4\textwidth]{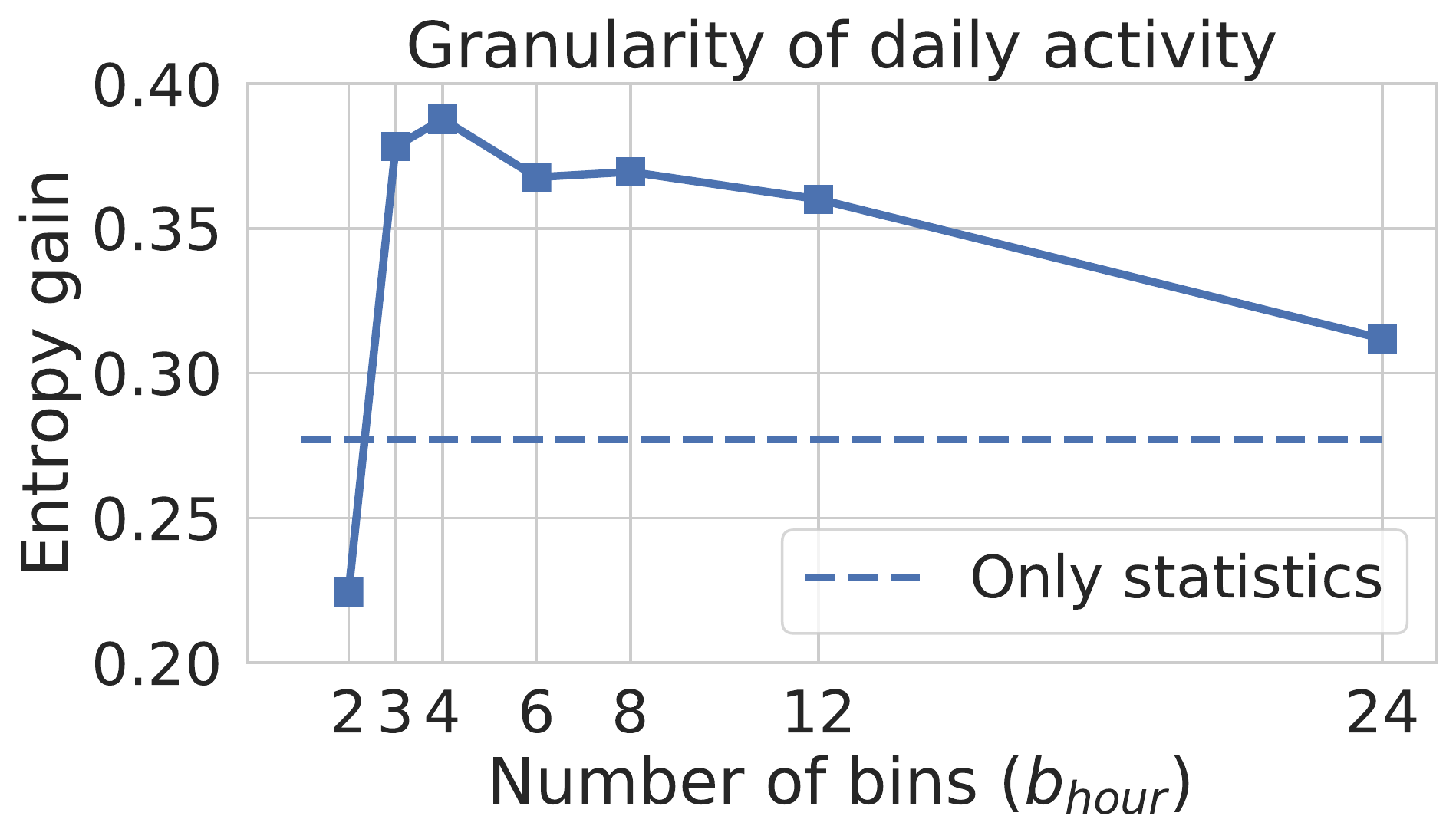}
    \includegraphics[width=0.4\textwidth]{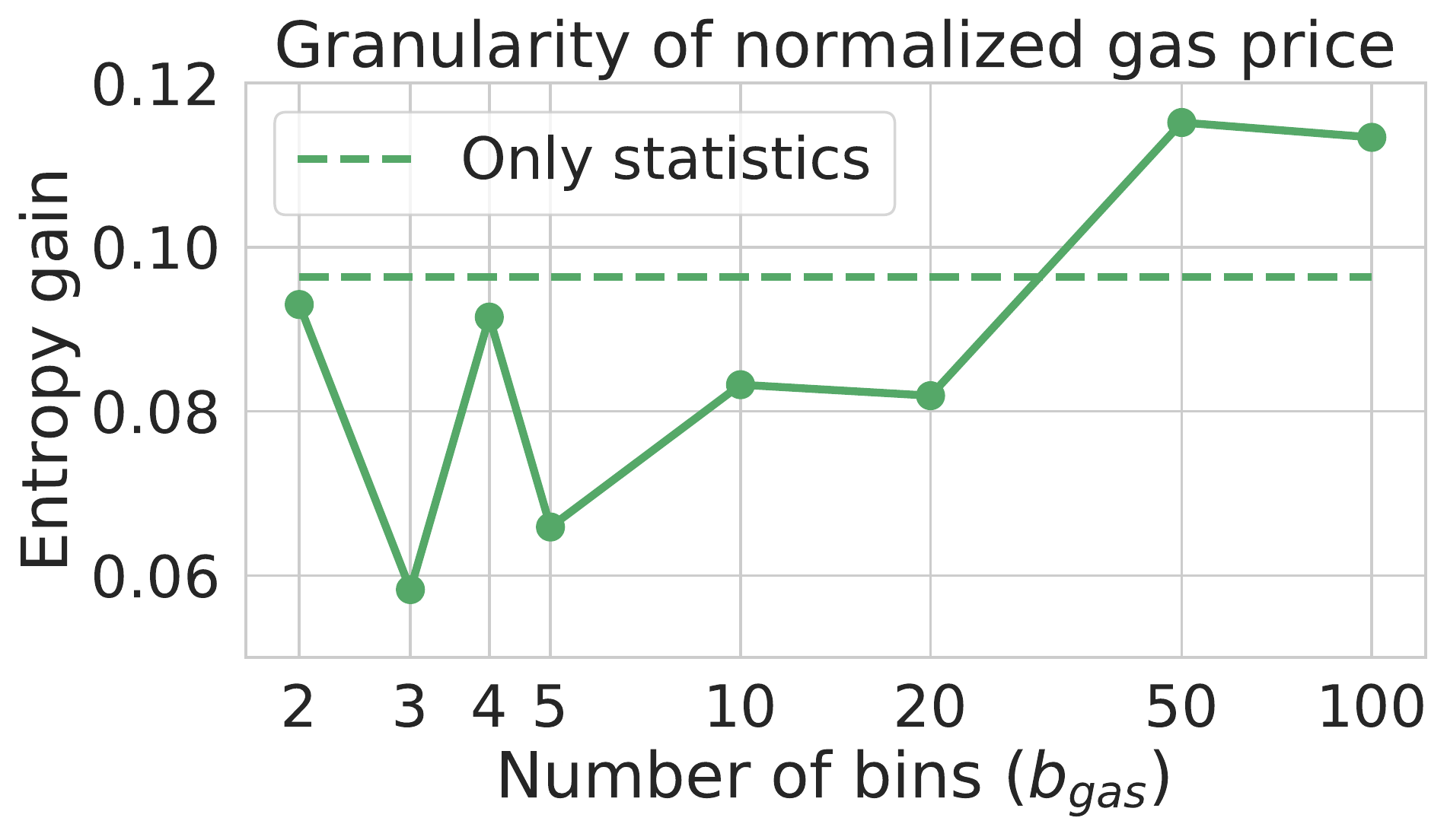}
    \caption{Entropy gain at different granularity for daily activity \textbf{(top)} and normalized gas price \textbf{(bottom)}. \textbf{Dashed lines} show performance with only mean, median and standard deviation used.}
    \label{fig:granularity-entropy}
\end{figure}

\begin{figure}
    \centering
    \includegraphics[width=0.4\textwidth]{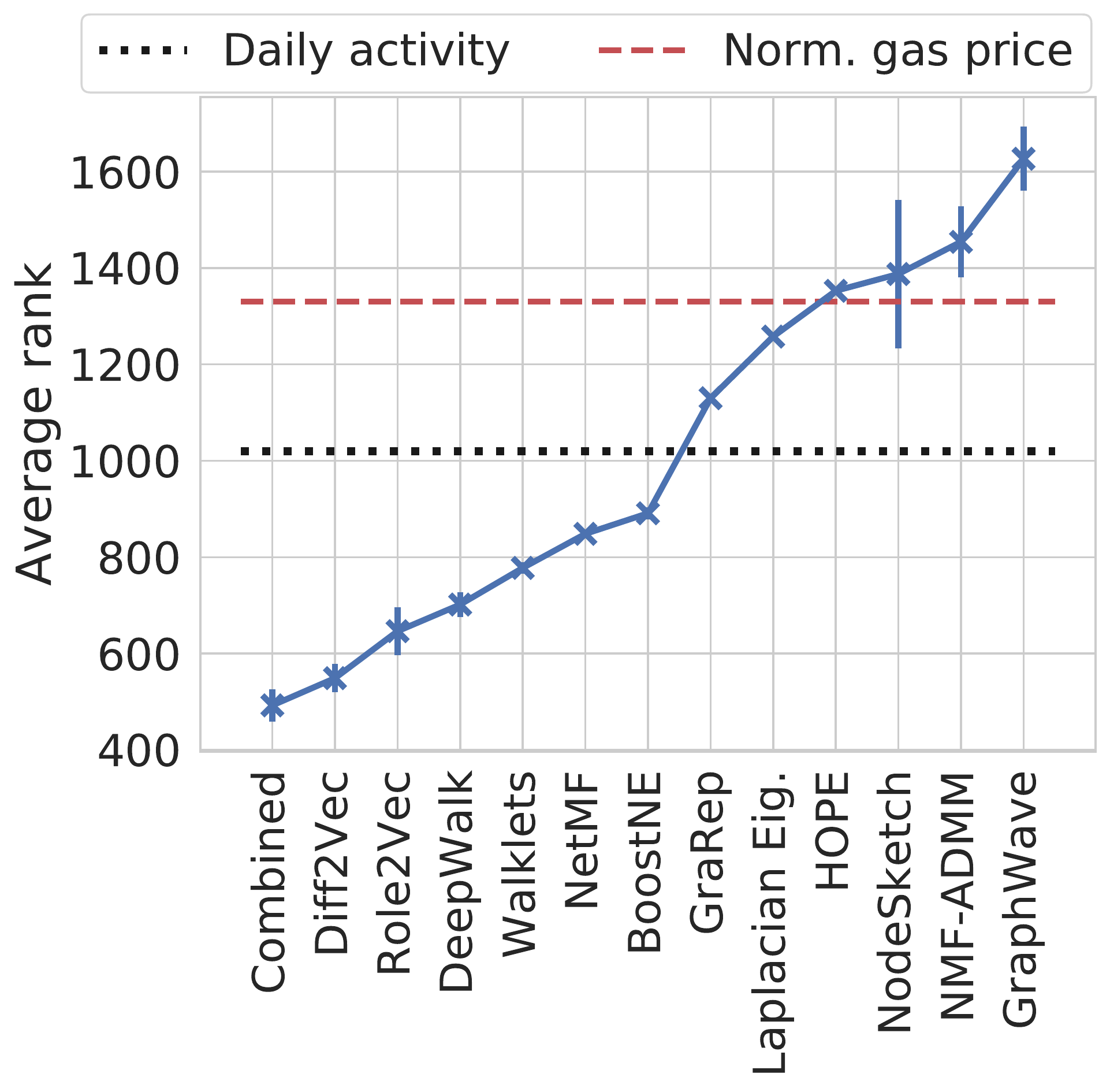}
    \caption{Average rank for node embedding methods. Vertical lines show  standard deviation in $10$ independent experiments. Reciprocal rank combination of Diff2Vec and Role2Vec gives the best performance. Note that the maximum rank is $3321$, the total number of Ethereum addresses considered in this experiment.}
    \label{fig:node_embeddings-avgrank}
\end{figure}

\begin{figure}
    \centering
    \includegraphics[width=0.4\textwidth]{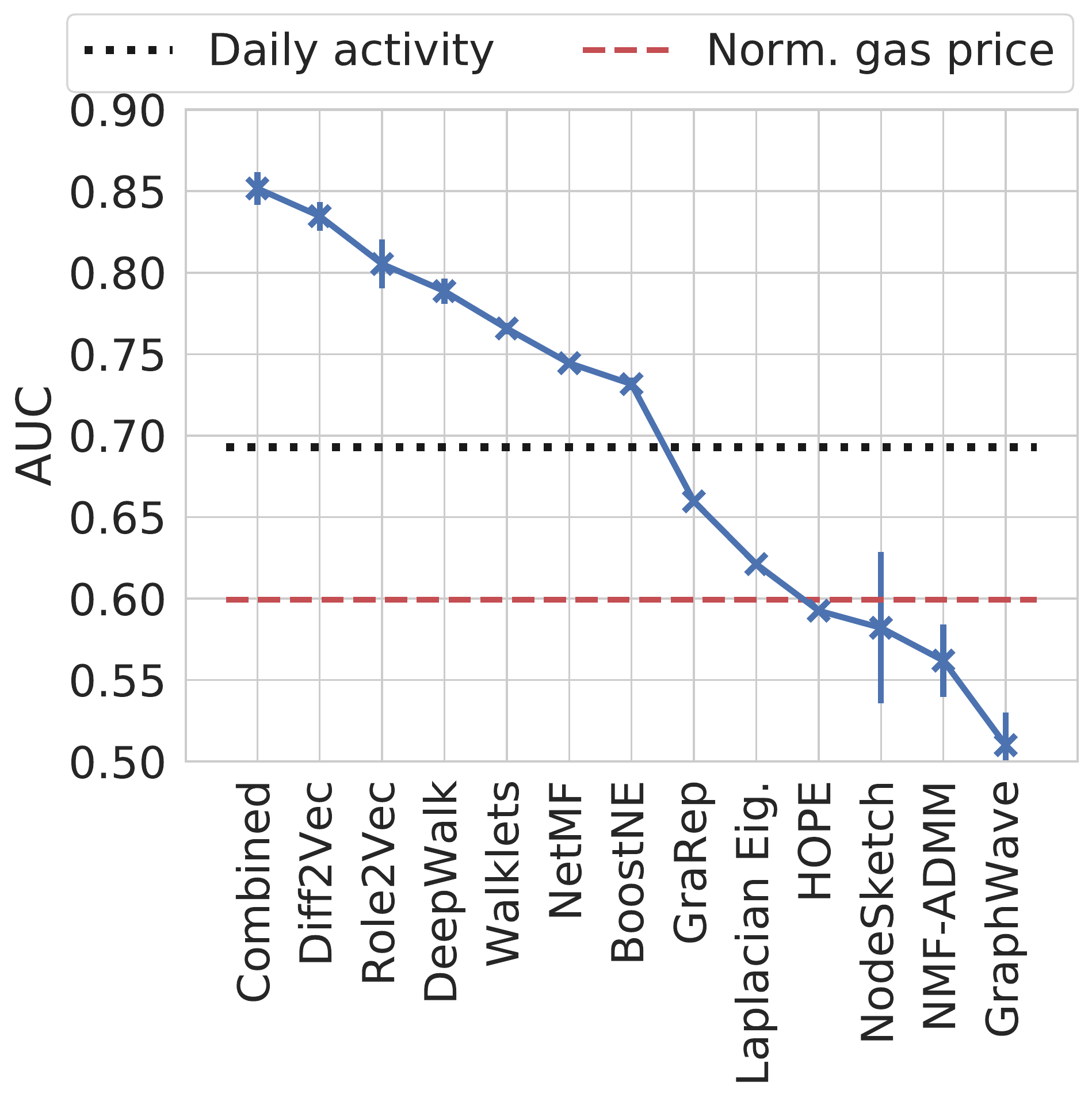}
    \caption{AUC for node embedding methods. Vertical lines show  standard deviation in $10$ independent experiments. Reciprocal rank combination of Diff2Vec and Role2Vec gives the best performance.}
    \label{fig:node_embeddings-auc}
\end{figure}

\begin{figure}
    \centering
    \includegraphics[width=0.4\textwidth]{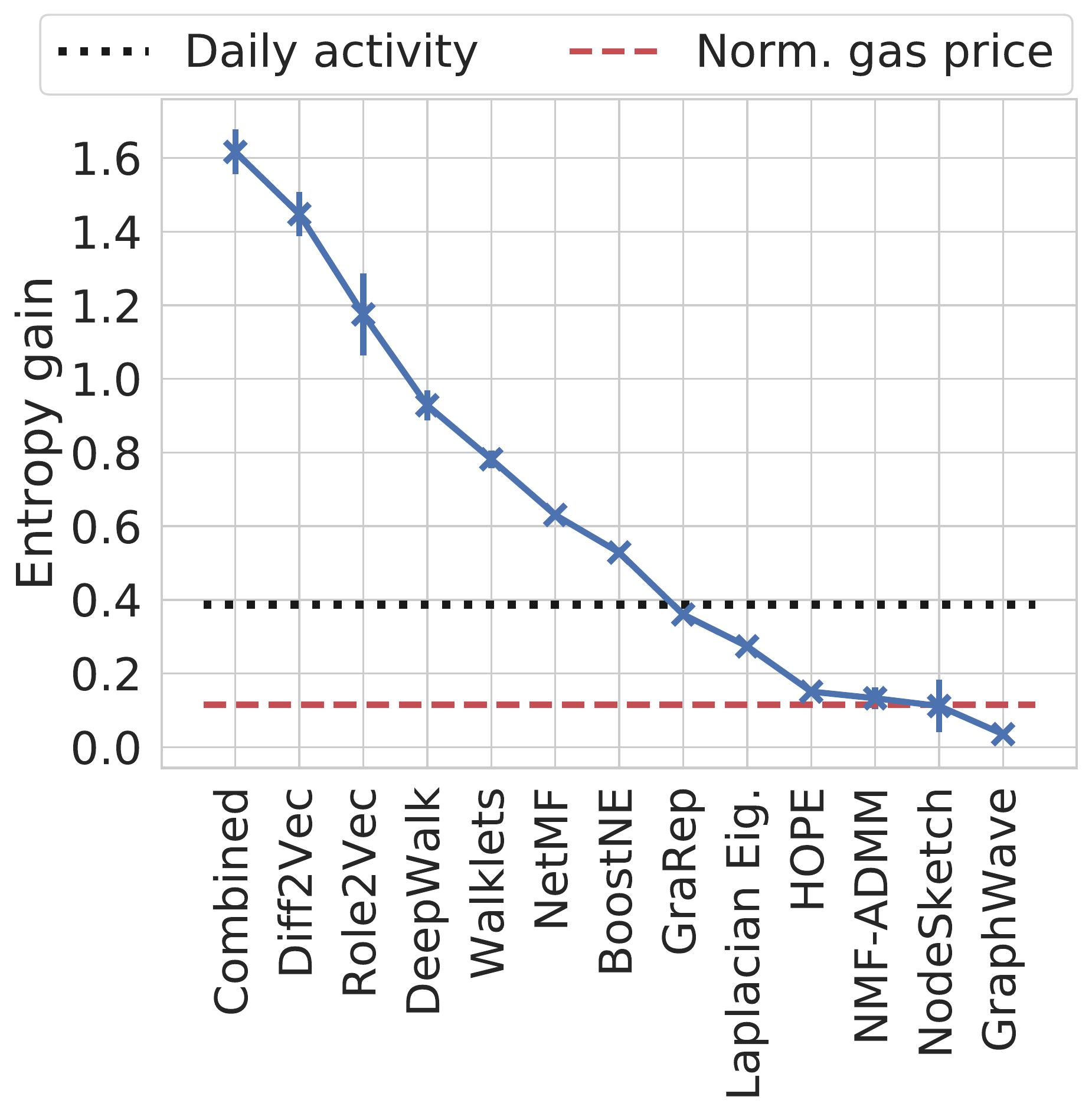}
    \caption{Entropy gain for node embedding methods. Vertical lines show  standard deviation in $10$ independent experiments. Reciprocal rank combination of Diff2Vec and Role2Vec gives the best performance.}
    \label{fig:node_embeddings-entropy}
\end{figure}

\begin{figure}
    \centering
    \includegraphics[width=0.4\textwidth]{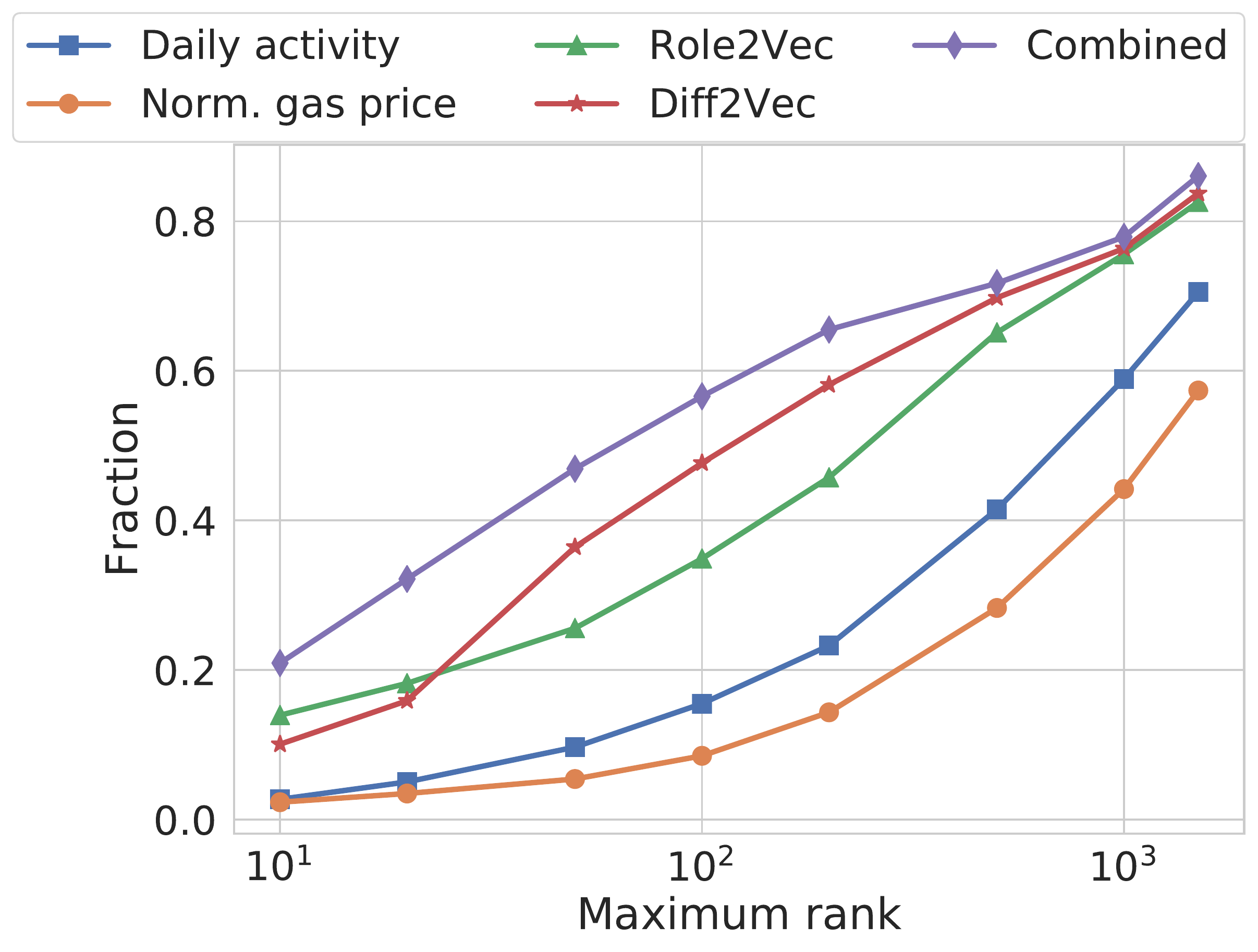}
    \caption{Fraction of ENS address pairs correctly identified within a given maximum rank, for different embedding methods.}
    \label{fig:ens_performance}
\end{figure}

\section{Deanonymizing trustless mixing services on Ethereum} \label{sec:deanonymising}
As the Ethereum community realises the consequences of the lack of privacy on Ethereum, more and more emphasis is put on increasing transaction privacy~\cite{meiklejohn2018mobius,seres2019mixeth,shlomovits2019sharelock}. Hence, privacy-enhancing tools became crucially important gadgets in the Ethereum ecosystem. Without doubt, the most popular is Tornado Cash (TC), a non-custodial zkSNARK-based mixer. It allows its users to enhance their anonymity by hiding their identity among a set of participating users. In this section, we provide techniques and heuristics to decrease the anonymity achieved in a TC mixer.

The Tornado Cash (TC) Mixers are sets of trustless Ethereum smart contracts allowing Ethereum users to enhance their anonymity. A TC mixer contract holds equal amounts of funds (ether or other ERC-20 tokens) from a set of depositors. One mixer contract typically holds one type of asset. In case of the TC mixer, anonymity is achieved by applying zkSNARKs~\cite{groth2016size}. Each depositor inserts a hash value in a Merkle-tree. Later, at withdraw time, each legitimate withdrawer can prove unlinkably with a zero-knowledge proof that they know the pre-image of a previously inserted hash leaf in the Merkle-tree. Subsequently, users can withdraw their asset from the mixer whenever they consider that the size of the anonymity set is satisfactory.

Cryptocurrency mixers typically provide $k$-anonymity (also known as plausible deniability) to their users~\cite{samarati1998protecting}. Generally speaking, a $k$-anonymized dataset has the property that each record is indistinguishable from at least $k-1$ others. Specifically, if a mixer contract holds $n$ deposits out of which $n-k$ had already been  withdrawn, then the next withdrawer will be indistinguishable among at least those $k$ users who have not withdrawn from the mixer yet. Hence each withdrawer can enhance their transaction privacy and make their identity indistinguishable among at least $k$ addresses. We call the set containing the $k$ indistinguishable addresses the anonymity set of the user.

\begin{figure}
    \centering
    \includegraphics[width=0.45\textwidth]{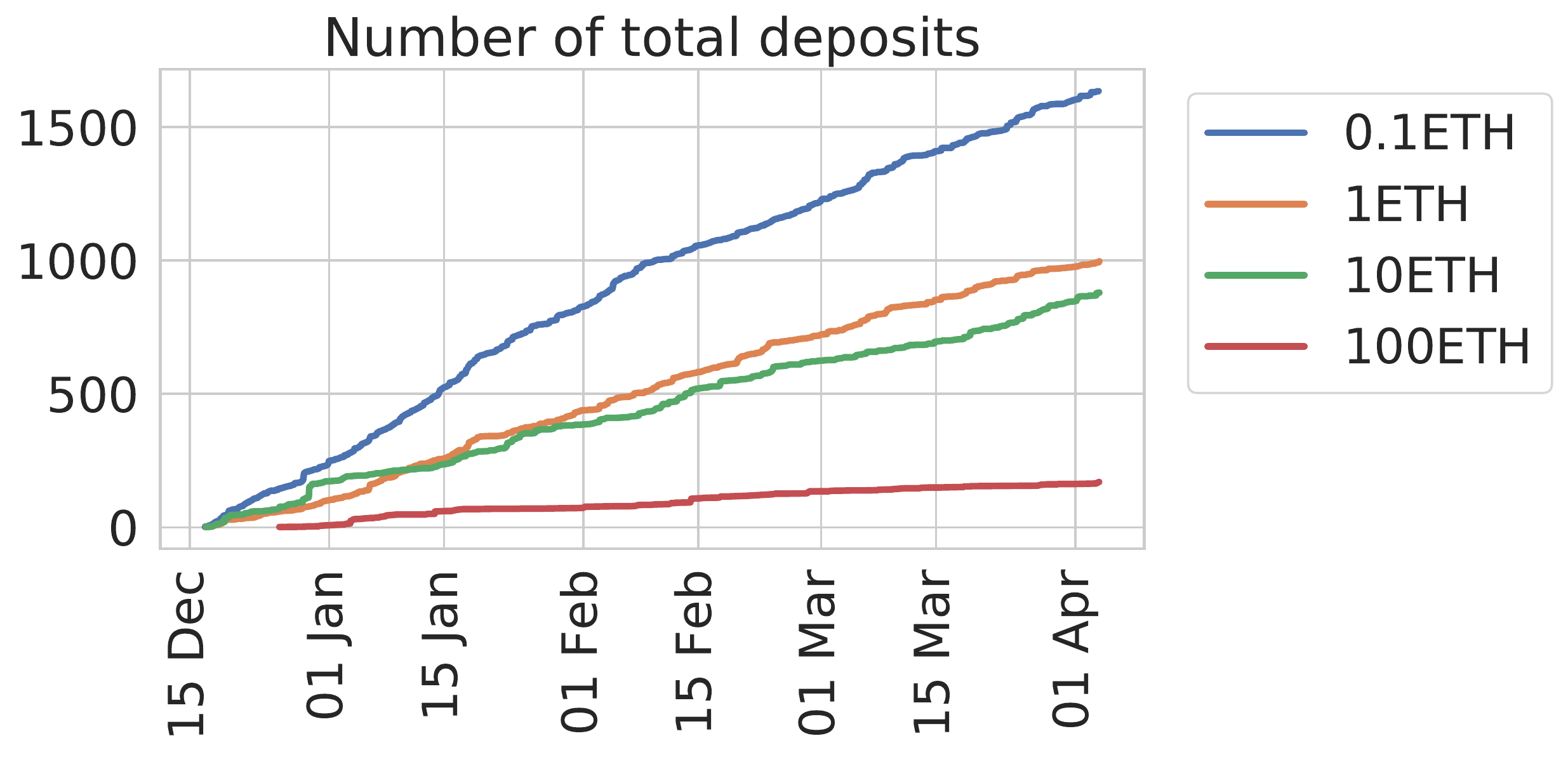}
    \caption{The number of total deposits in each TC mixer over time. This is an upper bound for the achievable anonymity set size when a withdraw transaction is executed. The popularity of the $0.1$ETH mixer is superior compared to higher value mixers.}
    \label{fig:deposits_over_time}
\end{figure}

In Figure~\ref{fig:deposits_over_time}, we show the changes in the anonymity set size over time for four TC mixer contracts ($0.1$ ETH, $1$ ETH, $10$ ETH, $100$ ETH) respectively. Since TC was launched in December 2019, hundreds of deposits were placed in the mixers as more and more user interacted with this service. In general, we observe orders of magnitude lower activity for the $100$ETH mixer, thus it does not provide as much anonymity as mixers with lower values ($0.1$ETH, $1$ETH, $10$ETH).

\subsection{Heuristics for linking mixer deposits and withdraws}\label{sec:heuristics}
Unfortunately, careless usage easily reveals links between deposits and withdraws and also impact the anonymity of other users, since if a deposit can be linked to a withdraw, it will no longer belong to the anonymity set.  Next, we list three usage patterns that can be used to link deposits and withdraws. The simplest careless usage is applying the same address for deposit and withdraw transactions as well:

\noindent \textbf{Heuristic 1.}
\textit{If there is an address from where a deposit and also a withdraw has been made, then we consider these deposits and withdraws linked.}

The next heuristic is based on salient gas price settings. Most wallet softwares, e.g. Metamask or My Ether Wallet, automatically sets gas prices as multiples of Gwei ($10^{9}$ wei, i.e. giga wei). However, one can observe gas prices whose last 9 digits are non-zero, hence those gas prices are likely set by the transaction issuer manually. These custom-set gas prices can be used to link deposits and withdraw transactions. For instance, one might observe the deposit transaction\footnote{Depositor:$\mathit{0x074a3e9451fe3fb47be47786cf2dc4e84e797a6f}$} at block height $9,418,956$ with $5.130909091$ Gwei gas price. Later on, there is a withdraw transaction\footnote{Withdrawer:$\mathit{0x0f2437ff38e032596f2226873038230dcb22c485}$} at block height $9,419,096$ in the Ethereum blockchain with exactly the same custom-set gas price. This deposit and withdraw pair can be linked. 

\noindent \textbf{Heuristic 2.} \textit{If there is a deposit-withdraw pair with \emph{unique and manually set gas prices}, then we consider them as linked.}

Frequently, users reveal links between their deposit and withdraw addresses if they sent transactions from one of their addresses to another address owned by them. We conjecture that users falsely expect that withdraw addresses are clean, therefore they can send transactions from any address to their clean withdraw addresses. However, if the withdraw address can be linked to one of their deposit addresses, then they effectively lose all privacy guarantee accomplished by the fresh withdraw address. Express differently, if users run out of clean funds at their fresh addresses, they might feel tempted to move "dirty" assets to their "clean" addresses. Again, such a transaction links "clean" and "dirty" addresses which is captured by the following heuristic.

\noindent \textbf{Heuristic 3.}
Let $d$ be a deposit and $w$ a withdraw address in a TC mixer. If there is a transaction between $d$ and $w$ (or vice versa), we consider  the addresses linked.

One could easily generalize Heuristic 3 by requiring transactions to be sent from not only a depositor address $d$, but rather from any address in the cluster of addresses containing $d$. However, we leave the implementation of this generalization for future work.

\begin{table*}
{\footnotesize
  \begin{center}
    \begin{tabular}{lccccc}\toprule 
    & \multicolumn{4}{c}{\textbf{Deanonymized withdraws}} & \textbf{All}\\
       \textbf{Mixer} & \textbf{Heuristic 1} & \textbf{Heuristic 2} & \textbf{Heuristic 3} & \textbf{Total}& \textbf{Withdraws}\\
      \hline
      $0.1$ETH & $\num[group-separator={,}]{95}$ $(7.5\%)$& $\num[group-separator={,}]{80}$ $(6.2\%)$& $113$ $(8.8\%)$&$218$ $(17.1\%)$&$\num[group-separator={,}]{1272}$\\
      $1$ETH & $21$ $(2.5\%)$& $40$ $(4.8\%)$& $75$ $(9\%)$&$110$ $(13.2\%)$&$833$\\
      $10$ETH & $8$ $(1.1\%)$& $9$ $(1.2\%)$& $46$ $(6.2\%)$&$60$ $(8.1\%)$&$738$\\
      $100$ETH & $2$ $(1.5\%)$& $5$ $(3.8\%)$& $3$ $(2.3\%)$&$7$ $(5.3\%)$&$132$\\
      \bottomrule
    \end{tabular}
  \end{center}
}
  \caption{Number of all withdraws and deanonymized withdraws using the corresponding heuristics in each mixer contract.}\label{tab:mixerheuristics}
\end{table*}

Applying Heuristics~1--3, we found $218$, $110$, $60$, and $7$ withdraws linked in the four mixer contracts ($0.1$ ETH, $1$ ETH, $10$ ETH, $100$ ETH) respectively up to 2020 April 4th, see Table~\ref{tab:mixerheuristics}. We note that withdraws identified by Heuristic 2 can also overlap with other withdraws identified by Heuristic 1 or 3. Hence the number of total linked withdraws are less than the sum of all withdraws individually identified by each heuristic. 

\subsection{Elapsed time between deposits and withdraws, withdraw address reuse}

\begin{figure}
    \includegraphics[width=0.4\textwidth]{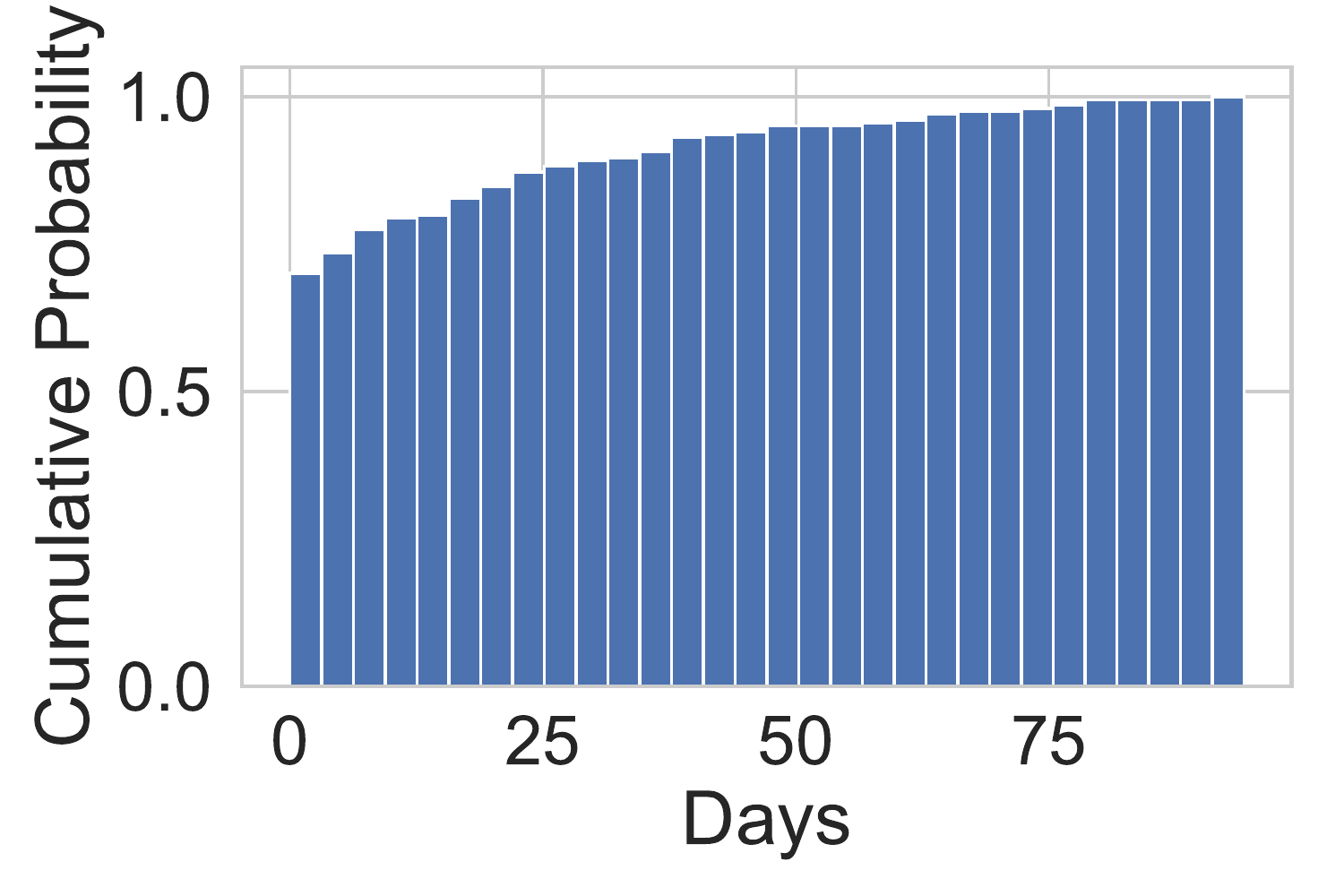}
    \caption{Elapsed time in days between linked  deposit and withdraw transactions for the $0.1$ ETH mixer contract. Vast majority of users do not wait more than one day to withdraw their deposits.}
    \label{fig:mixingintervals}
\end{figure}

\begin{figure}[h]
    \centering
    \includegraphics[width=0.4\textwidth]{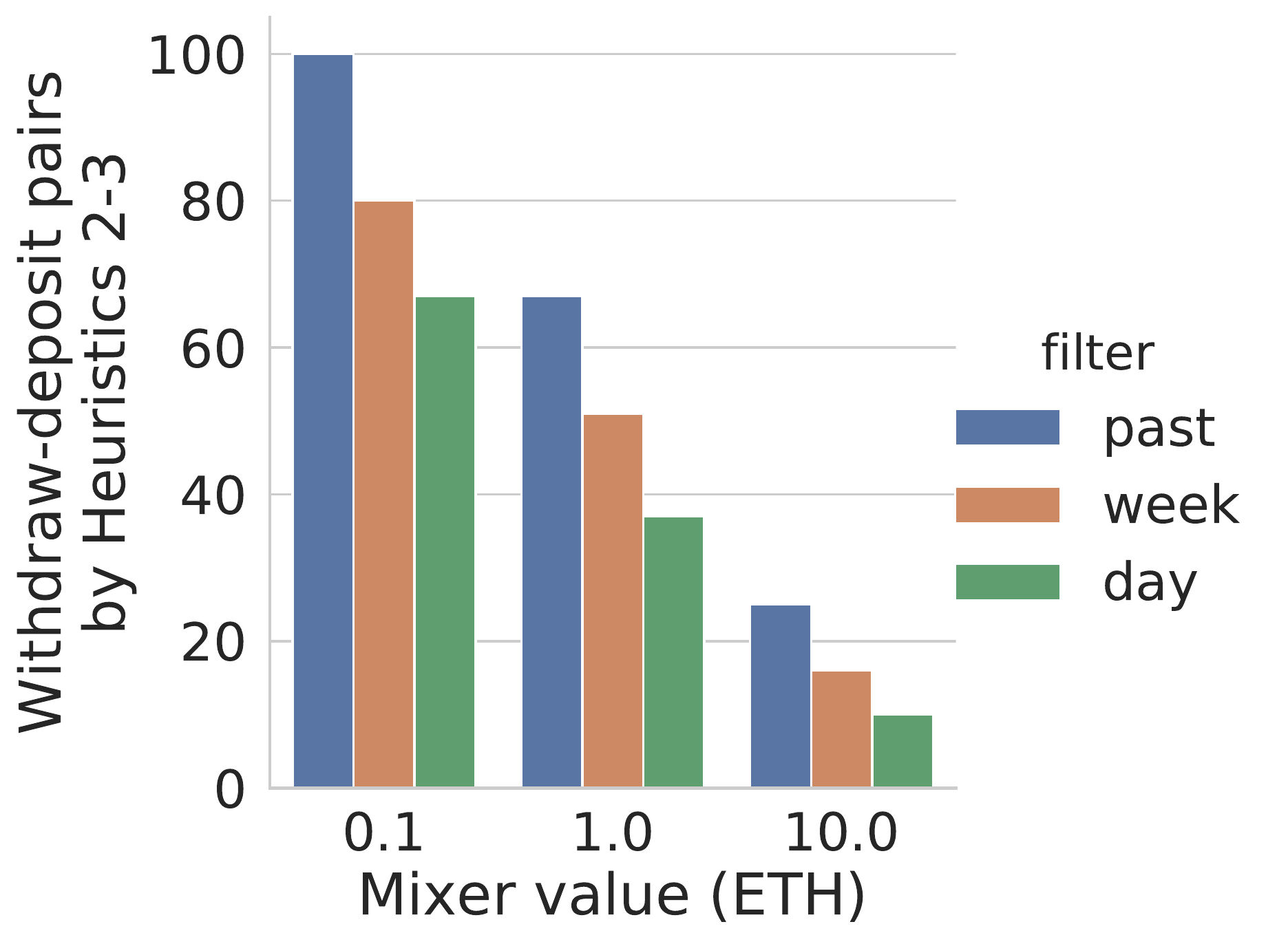}
    \caption{For each mixer, the number of withdraw-deposit pairs linked by Heuristics~2--3 such that the deposit is not later than the day or the week before, or any time in the past.}
    \label{fig:heuristics_in_eval}
\end{figure}

In Figure~\ref{fig:mixingintervals}, we observe that most users of the linked deposit-withdraw pairs leave their deposit for less than a day in the mixer contract. This user behavior can be exploited for deanonymization by assuming that the vast majority of the deposits are always withdrawn after one or two days. 


\begin{figure}
    \includegraphics[width=0.4\textwidth]{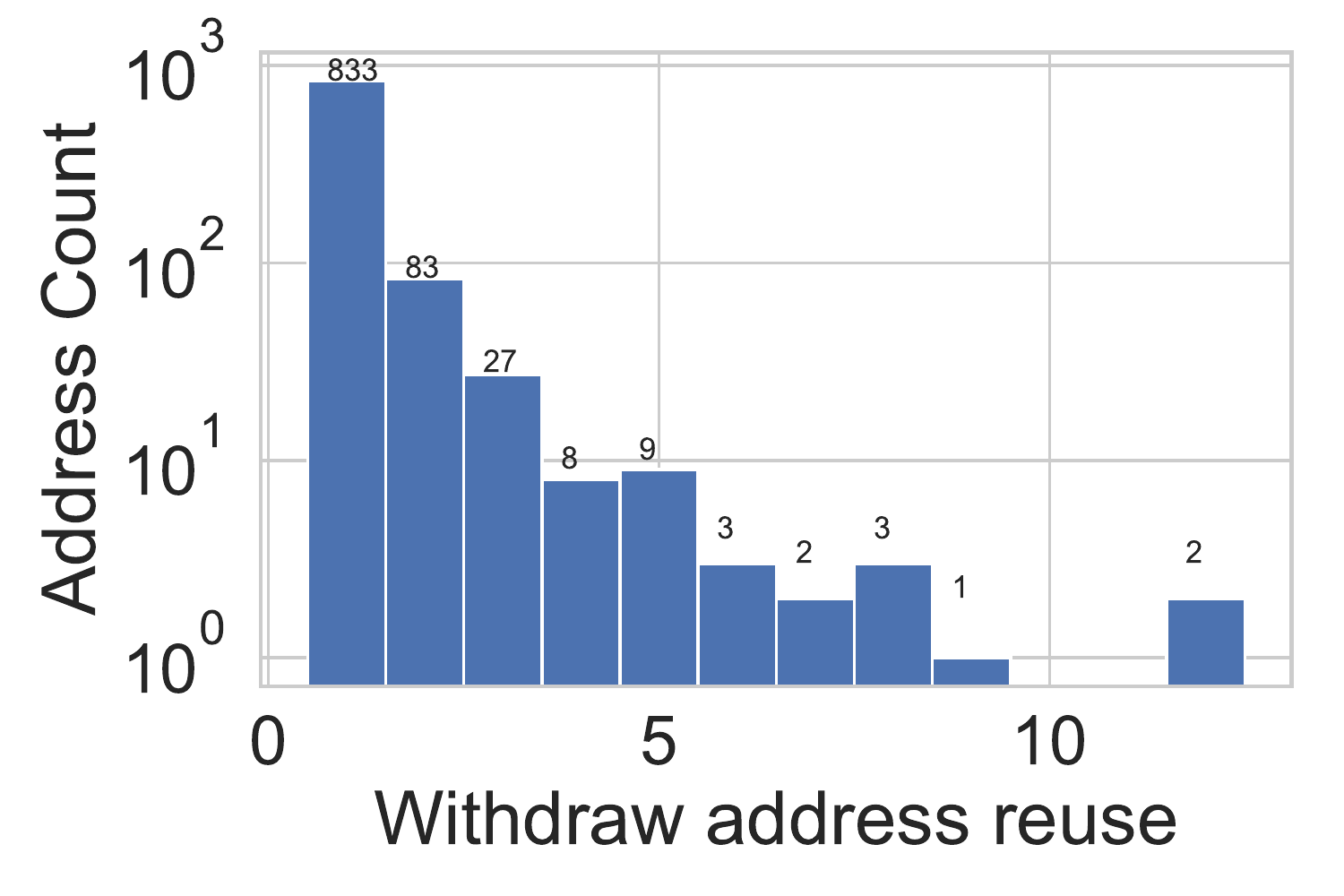}
    \caption{Withdrawal address reuse in the $0.1$ ETH mixer contract. Many users withdraw multiple deposits to the same address, which eases deanonymization and reduces the privacy properties of the mixer.}
    \label{fig:withdrawclusters}
\end{figure}

Even worse, in Figure~\ref{fig:withdrawclusters}, we observe several addresses receiving multiple withdrawals from the $0.1$ ETH mixer contract. For instance, there are $83$ addresses that have withdrawn 2 times and $27$ addresses  with 3 withdrawals each. This phenomenon causes privacy risk not just for the owner of these addresses but also reduces the privacy properties of the mixer. Note that proper usage always requires a withdraw to a fresh address.

\begin{figure*}[p]
    \centering
    \includegraphics[width=\textwidth]{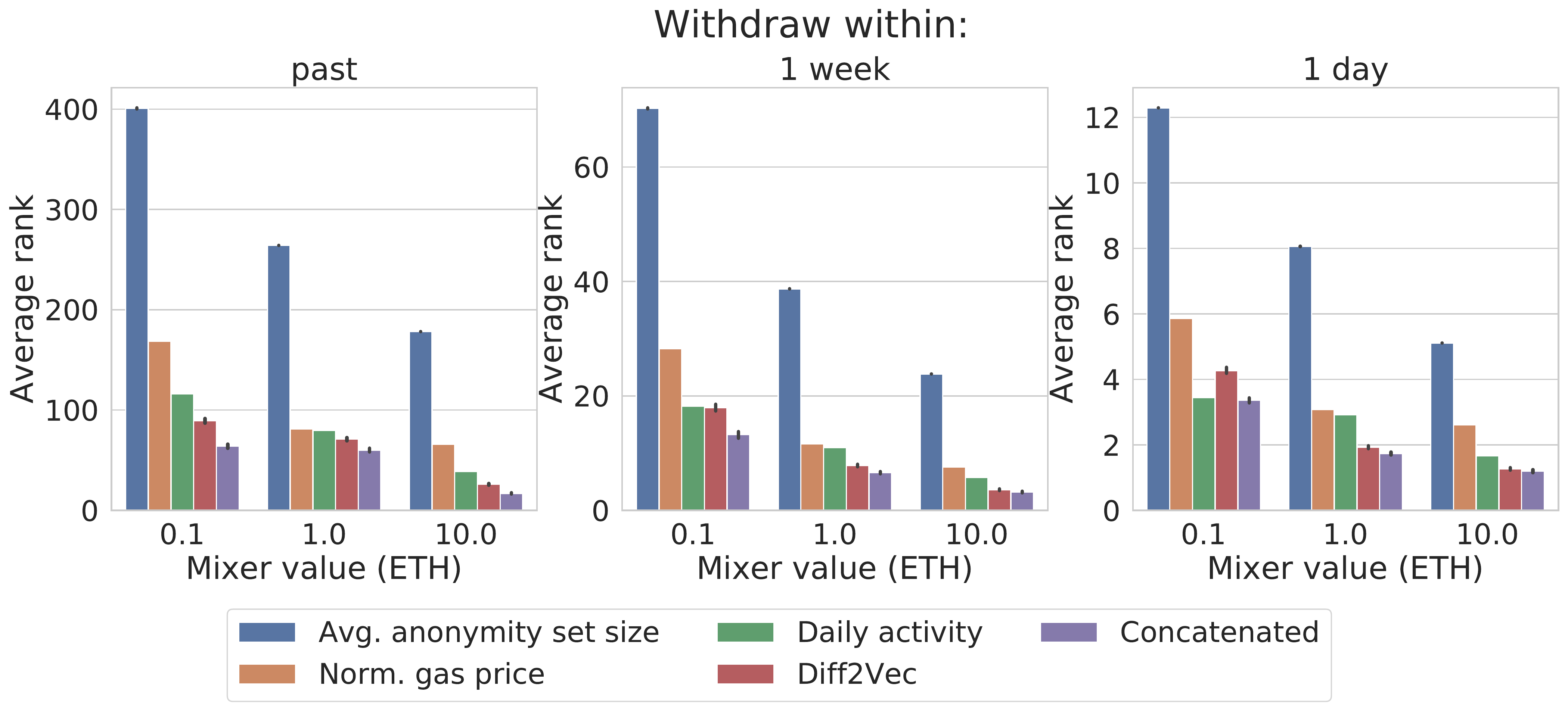}
    \caption{Average rank of the deposit address in the candidate list of our algorithms for the three different ground truth sets described in Section~\ref{sec:tornado_performance}.}
    \label{fig:tornado_performance}
\end{figure*}

\begin{figure*}[p]
    \centering
    \includegraphics[width=\textwidth]{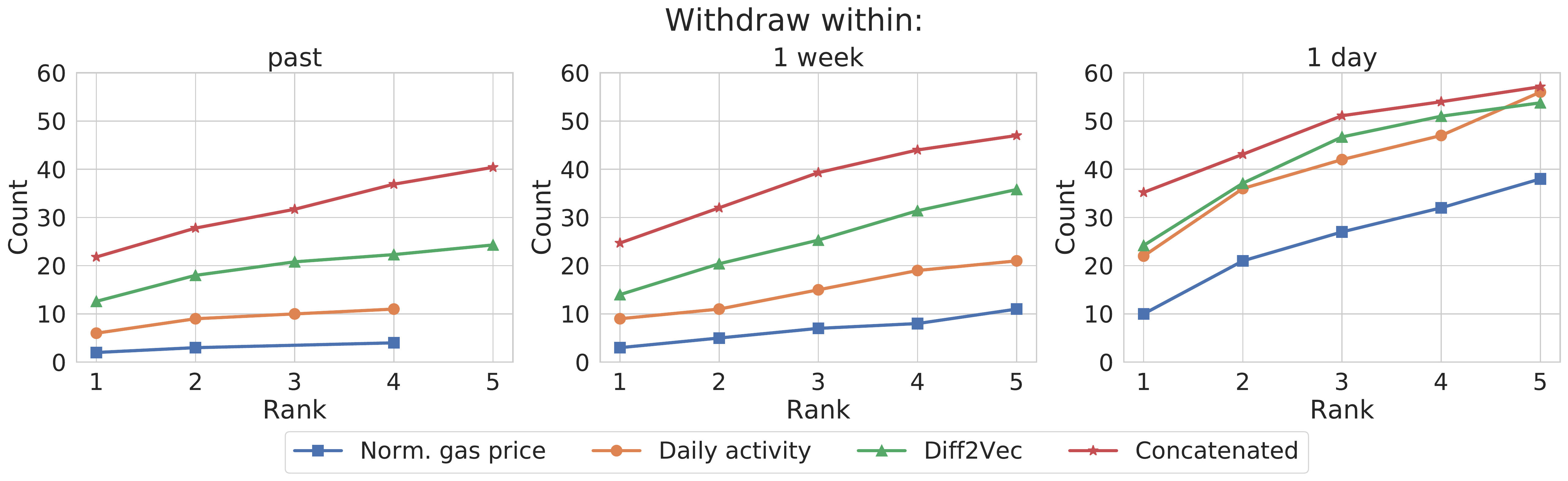}
    \caption{Number of withdraw addresses in the $0.1$ETH mixer contract such that the corresponding deposit is identified within the given rank in the candidate list of each deanonymization technique, separate for the three ground truth sets described in Section~\ref{sec:tornado_performance}.}
    \label{fig:tornado_rank}
\end{figure*}

\begin{figure*}[p]
    \centering
    \includegraphics[width=\textwidth]{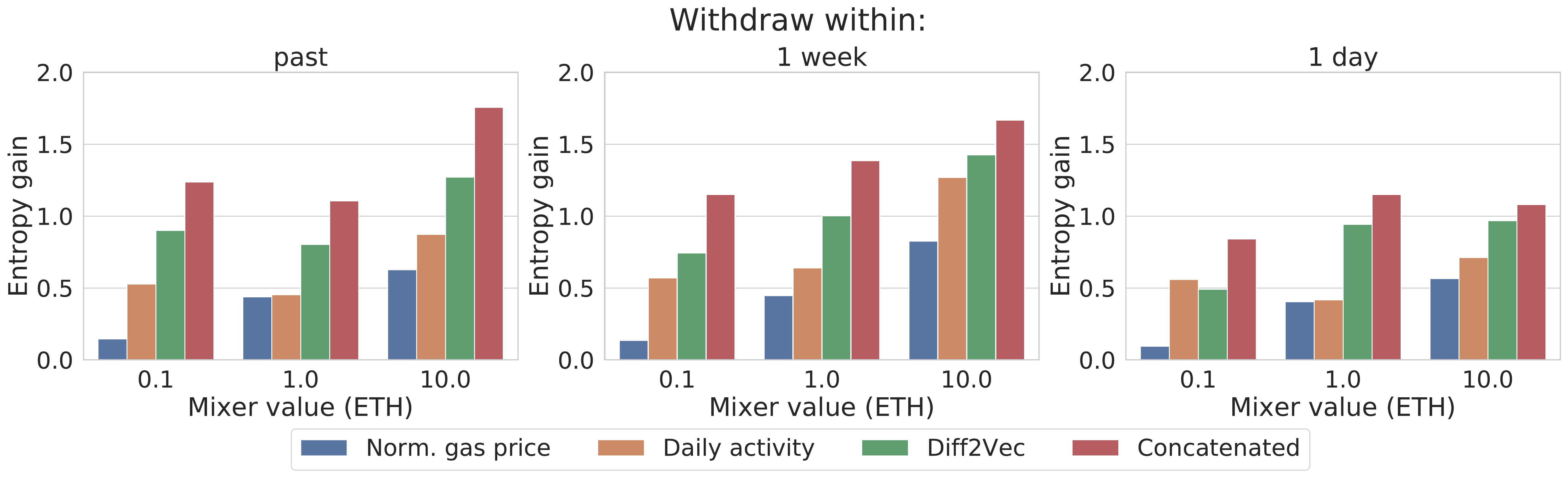}
    \caption{Entropy gain of our best deanonymization methods for the three different ground truth sets described in Section~\ref{sec:tornado_performance}.}
    \label{fig:tornado_entropy_gain}
\end{figure*}

\subsection{Deanonymization performance}\label{sec:tornado_performance}

Next we measure how well the techniques of Section~\ref{sec:pairing} identify the linked withdraw-deposit address pairs. 
We build ground truth by using Heuristics~2--3 of Section~\ref{sec:heuristics}. We omit withdraw-deposit pairs identified by Heuristic~1 from the ground truth as in that case both withdraw and deposit addresses are the same. Such a ``pair'' is trivially identified. We define three different \textbf{ground truth sets}, one when the deposit is within the past day of the withdraw, another when within the past week, and the unfiltered full set, see Fig.~\ref{fig:heuristics_in_eval}. Experiments on the unfiltered full set is labeled \emph{past} in Figures~\ref{fig:tornado_performance}-\ref{fig:tornado_temporal}.

Note that our ground truth sets are compiled by using Heuristics~2--3, and hence are correct up to our best knowledge on the data. Since in Heuristic~2 we used gas prices and in Heuristic~3 an edge between the two addresses, in this section, we show gas price only as reference, and omit the edges used by Heuristic~3 for the network analysis algorithms. As we will see, gas price distribution performs weakly for finding the account pairs identified by the Heuristics despite that Heuristic~2 is based on gas price, adding the edges between accounts identified by Heuristic~3 would yield overly strong deanonymization results since the same information is used for deanonymization and testing. 

Figure~\ref{fig:tornado_performance} shows that an address with withdraw within a day or week has significantly smaller anonymity set size, on  average, since we only search for the corresponding deposit in a smaller set. 
For example, for the $0.1$ETH mixer the original average anonymity set size of $400$ could be reduced to almost $12$ by assuming that the deposit occurred within one day of the withdraw. 

We note that in Figure~\ref{fig:tornado_performance} and all other measurements over the filtered ground truth sets, we do not discount for the withdraw addresses that are not included in the filtered set.  For example, as seen in Figure~\ref{fig:mixingintervals}, for 80 0.1-Ether withdraw transactions, we list candidate deposits, but for the remaining 20, we make no deanonymization attempt.  To normalize the results by considering these withdraws, we have to assume that the corresponding deposit is not in the 80-element candidate set but in the remaining 320, thus giving an average rank contribution of 160 for 20\% of the data.  Hence average rank for 0.1-Ether withdraws with deposit within a week have an additional correction of 32 for average rank; by similar calculations, the correction for transactions within a day is 63.

Daily activity and Diff2Vec have similar performance while their concatenated feature vectors proved to be the best address representation; for the smaller ground truth sets, they identify related deposit addresses within the 20 and 5 closest representations on average. Withdraw linking performance is further improved by concatenating the two models.  Entropy gain is shown in Figure~\ref{fig:tornado_entropy_gain} and the number of withdraws linked to deposits within a given rank of the output for the best methods are in Figure~\ref{fig:tornado_rank}.

In Figure~\ref{fig:tornado_temporal}, we show the withdraw linking performance over time. As the number of active deposits increases, it becomes harder to link withdraws to any of the past deposits. However withdraws that follow the deposit after a few days are still much easier to deanonymize. 

\begin{figure*}
    \centering
    \includegraphics[width=\textwidth]{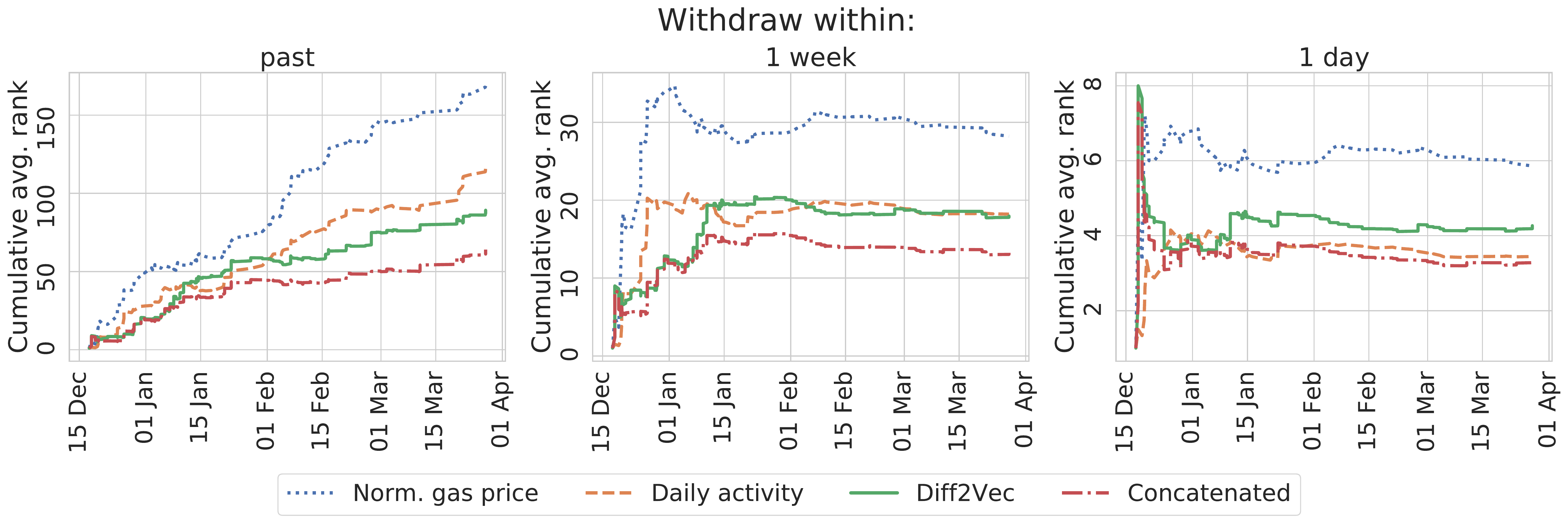}
    \caption{Change of average rank in time, cumulated from the beginning of our data, for the $0.1$ ETH Tornado mixer by using our best deanonymization methods. Results are showed separately for the three ground truth sets described in Section~\ref{sec:tornado_performance}.}
    \label{fig:tornado_temporal}
\end{figure*}

\subsection{Maintaining privacy}
We believe if users were using the technology in a sound way or a privacy-focused wallet software would have helped them and abstracted away potential privacy leaks, then TC mixers could possibly achieve higher degrees of anonymity.
\subsubsection{Randomized mixing intervals}
Mixing participants decrease largely their gained anonymity by withdrawing funds after short time intervals, cf. Figure~\ref{fig:mixingintervals} and~\ref{fig:tornado_performance}. These heuristics can be defeated by randomized mixing intervals. Randomized mixing intervals cannot be enforced by the mixing contract itself, since withdrawals are unlinkable to the deposits. Therefore, this should be accomplished by the user wallet software.
\subsubsection{Fresh withdraw addresses}
Currently, many users apply the same withdraw addresses across several withdraws, see Figure~\ref{fig:withdrawclusters}. This greatly decreases the complexity of linking deposits and withdraws. Therefore users must use fresh withdraw addresses for each of their withdraws. This issue could have been easily fixed on the user interface level.
\subsubsection{Mixer usage and user behaviors}
Mixers mainly attempt to break the link between sets of transaction  graphs associated with Ethereum accounts. As such, users need to ensure that their on-chain behaviors are unlinkable between uses of the TC mixers. Therefore, to ensure maximal privacy, users should use the TC mixers after every transaction. However, this decreases the user experience and ability to use applications on Ethereum.

\section{Danaan-gift attack in Ethereum}\label{sec:danaangift}
The Danaan-gift attack, also known as malicious value fingerprinting, was introduced in~\cite{biryukov2019privacy}. In a value fingerprinting attack, an adversary sends a cryptocurrency transaction with a crafted amount to add a fingerprint to the receiver's account balance. Although value fingerprinting was originally introduced in the context of Zcash, we notice that these attacks are applicable to Ethereum as well. Most wallet software denominates gas prices in multiples of gwei ($10^9$ wei where $1\mathit{ETH} =10^{18}\mathit{wei}$), hence transaction fees overwhelmingly (in $98,1\%$) do not change the last $9$ digits of an account balance. Albeit,  users might set transaction fees manually, potentially changing their own fingerprint (in $1.9\%$). The last $9$ digits of an account balance have no economic significance (1 gwei$\approx 0.0000003\$) $ but could be used as a fingerprint by an adversary.

\begin{table*}
  \begin{center}
    \begin{tabular}{lccccc}\toprule 
       \textbf{Tx} & \textbf{Addresses} & \textbf{Txs} &\textbf{Txs} & \textbf{Avg. Sent} & \textbf{Fingerprint}\\
       \textbf{Cutoff} & & & \textbf{Fingerprinting}&\textbf{Txs/Address} & \textbf{survival prob.}\\
      \hline
      50 & $\num[group-separator={,}]{56399}$& $\num[group-separator={,}]{120461}$&$\num[group-separator={,}]{61393}$& $2.14$&$21.83\%$\\
      100 & $\num[group-separator={,}]{56973}$& $\num[group-separator={,}]{161427}$&$\num[group-separator={,}]{73340}$& $2.83$&$17.97\%$\\
      500 & $\num[group-separator={,}]{57951}$& $\num[group-separator={,}]{384369}$&$\num[group-separator={,}]{129431}$& $19.48$&$6.56\%$\\
      All & $\num[group-separator={,}]{58367}$& $\num[group-separator={,}]{1137558}$&$\num[group-separator={,}]{352042}$& $19.49$&$0.073\%$\\
      \bottomrule
    \end{tabular}
    
  \end{center}
  \caption{Balance fingerprinting statistics for Ethereum users. In each cutoff, we only consider addresses that did not issue more transactions than the cutoff value. We observe that vast majority of fingerprinting transactions were sent by addresses that send numerous transactions. Fingerprinting an address with few sent transactions is obviously easier than an address with many issued transactions. Fingerprint survival probabilities were calculated as in Equation~\ref{eq:survivalprob}. }\label{tab:fingerprintingstatistics}
\end{table*}

\begin{figure}[t]
    \includegraphics[width=0.5\textwidth]{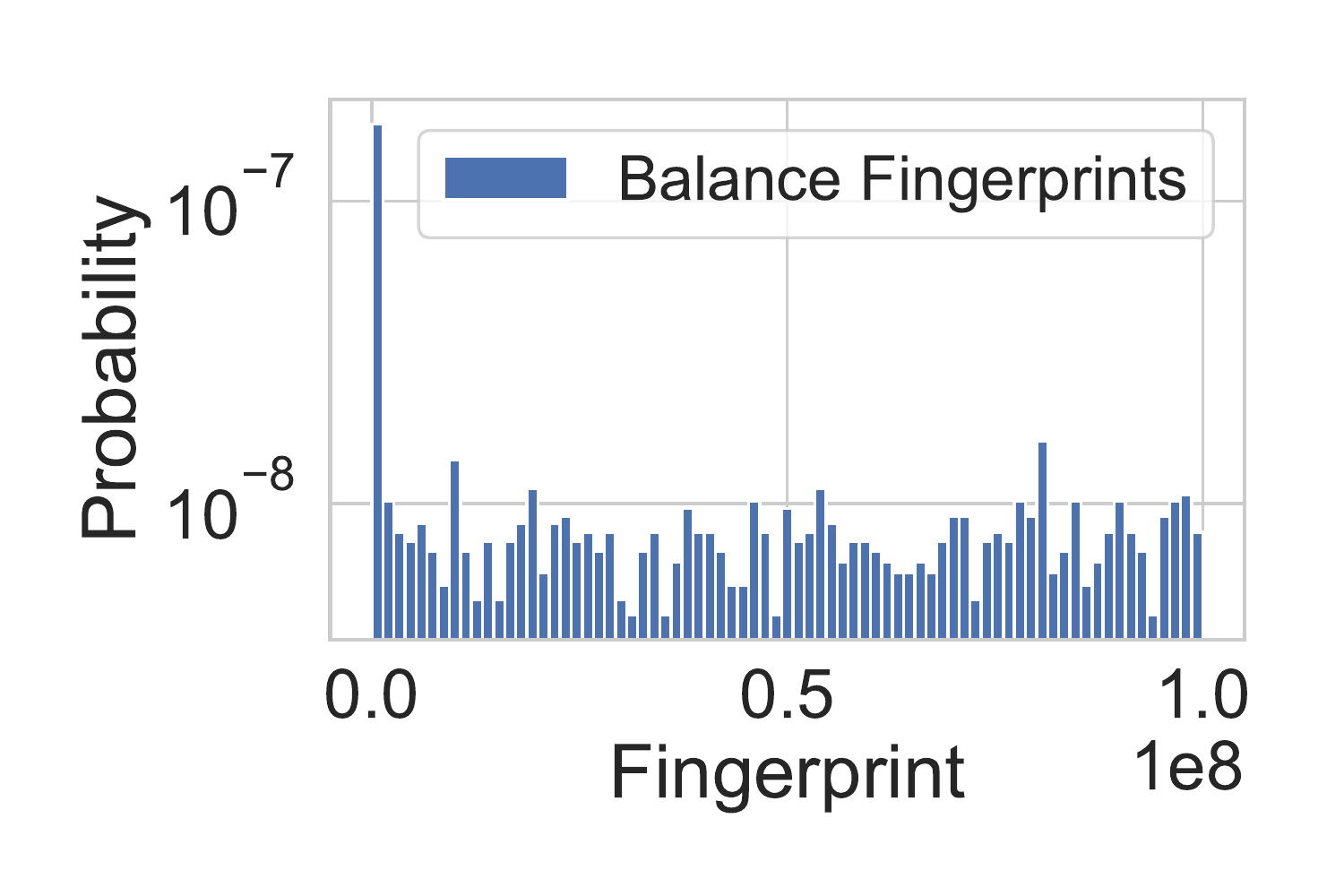} 
    \caption{Ether account balance fingerprints. Many Ethereum accounts have an integer account balance. This allows an attacker to fingerprint the last $9$ digits of an account balance. Account balance fingerprints distribution has a $4.01$ bit entropy and $6.44$ bit entropy gain.}
    \label{fig:allfingerprints}
\end{figure}

\begin{figure}[t]
\centering
\includegraphics[width=\linewidth,trim={11.5cm 15cm 9.5cm 4.2cm},clip]{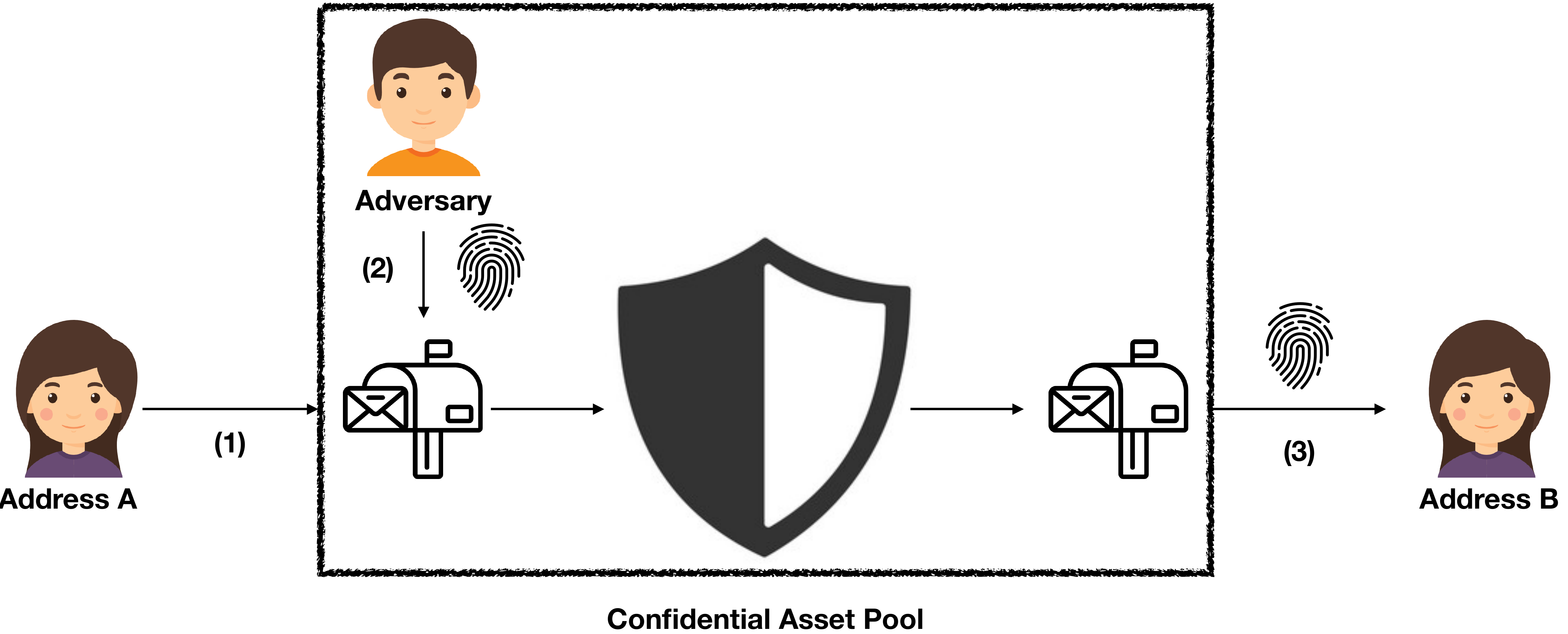}
\caption{Danaan-gift attack in confidential transaction layers. An adversary can fingerprint (2) an unsuspecting user's account balance after she deposited assets (1) in a confidential asset pool, e.g. AZTEC. Adversary can track the user when she leaves (3) the confidential pool.}
\label{fig:danaanexplainer}
\end{figure}

First, we measure the fraction of ether transfer transactions that modify the account fingerprint ($43,7\%$). 
For the sake of robustness of the measurements, we chose fingerprints with the last eight digits. As seen in Figure~\ref{fig:allfingerprints}, account balances are mostly integer values. However, the rest of the fingerprint values modulo $\num[group-separator={,}]{100000000}$ are moderately uniformly distributed. The entropy of the account balance fingerprints is $4.01$ with a $6.44$ entropy gain. These results suggest that account balances might be easily fingerprinted. In the sequel, we  estimate the average  fingerprint survival probability.


Let $F$ denote the event that a fingerprint of an address remains unchanged. To approximate the event probability  $\Pr(F)$, let $p$ denote the probability that a transaction modified the fingerprint and let $x$ denote the number of transactions sent or received by the given address in our dataset. By assuming that each transaction is independent from all others, the fingerprint survival probability of this address is $(1-p)^{x}$.

We observe that the distribution of the number $x$ of transactions sent and received by an address follow power-law distribution $\sim x^{-k}$ with $k=1.91$. 
The average survival probability of all addresses can hence be approximated by the following integral, where we group by $x$, the number of transactions of an address:
\begin{equation} \label{eq:survivalprob}
    \Pr(F)=\int_{1}^{\infty} x^{-k}(1-p)^{x} dx,
\end{equation}
which can be computed in a closed formula. The numerical values are summarized in Table~\ref{tab:fingerprintingstatistics}.

As the number of transactions sent follow a power-law distribution, the average value  is skewed by the tail of the distribution. Therefore it makes sense to calculate the average survival probability for several cutoffs of the tail, see Table~\ref{tab:fingerprintingstatistics}. Namely, in each cutoff we only consider addresses in our data set that sent less number of transactions than the cutoff value. One can observe how fingerprint survival probability increases among users with a small number of transactions.  For example, an adversary could successfully fingerprint $21.83\%$ of the addresses that send not more than $50$ transactions. This result is comparable to the $16.6\%$ fingerprint survival probability observed in Zcash~\cite{biryukov2019privacy}. 

\subsection{Danaan-gift attack for confidential transaction overlays}
A future application of Danaan-gift attacks in Ethereum might be linking confidential transactions in privacy-enhancing overlays, like the AZTEC protocol~\cite{williamson2018aztec}.

In a confidential transaction overlay, users can convert public amounts to confidential notes. Subsequently, they can send confidential notes to intended recipients by splitting and or joining their confidential notes. The amount of confidential notes is hidden, yet publicly verifiable due to range proofs. Users can also convert their confidential tokens back to public amounts. 

In this scenario, an adversary can fingerprint unsuspecting users inside a confidential transaction overlay, see Figure~\ref{fig:danaanexplainer}. When a user deposits a public amount to the confidential asset pool, an adversary could fingerprint her account balance by sending her a confidential transaction with a fingerprinting amount. Subsequently, the user might issue several confidential transactions in this privacy-enhanced overlay. If the victim's balance fingerprint survives during the course of issued confidential transactions, the adversary can identify the user withdrawing funds from the confidential asset pool by inspecting the fingerprint on the withdrawn amount. Thus the fingerprinting adversary can assess how much the unsuspecting user paid in the confidential asset pool.


\section{Future directions}\label{sec:futurework}
We expect that in the near future more potent and powerful deanonymization tools and techniques will emerge. In this work, we solely applied on-chain data for deanonymizing Ethereum users. Subsequent tools will likely use a combination of on-chain and off-chain data. Therefore we deem the following directions would be extremely valuable for future work for the broader cryptocurrency research community.

\subsection{Further quasi-identifiers}
In this work we identified several quasi-identifiers of Ethereum accounts, such as time-of-day activity, gas price profile and position in the Ethereum transaction graph. However, we forecast that many more quasi-identifiers can be used for further profiling and deanonymizing Ethereum users. One such potential quasi-identifier is wallet fingerprints. One could establish which wallet a certain user employs by assessing how transaction gas prices are calculated. Different wallet softwares use different methods to compute suggested gas prices~\cite{werner2020step}.

\subsection{Network-level privacy}
Assessing Ethereum's privacy provisions entirely can only be established if one considers the full life-cycle of a transaction. Specifically, one also needs to understand how much privacy is lost when users interact with full nodes or wallet providers. 

As the history of Bitcoin and other cryptocurrencies showed, full nodes and wallet providers can deanonymize regular users and light clients already on the network layer~\cite{biryukov2018deanonymization,biryukov2019security,fanti2017anonymity,fanti2017deanonymization,tramer2020remote,moser2018empirical}. An attacker could establish many well-connected nodes in the peer-to-peer layer to log the timing information of transactions. Due to the symmetry of broadcast, the adversary could infer the origin of the transaction~\cite{fanti2017deanonymization,biryukov2018deanonymization}. Yet, there are solely measurement studies on Ethereum's P2P network structure~\cite{kim2018measuring,gencer2018decentralization}. Therefore, it would be worthwhile to conduct a study on Ethereum's P2P network, but from a privacy point of view. Fortunately, several proposals had been made to enhance network-level privacy for cryptocurrencies~\cite{bojja2017dandelion,fanti2018dandelion++}.

Additionally, in Ethereum, special nodes called relayers gain more and more popularity. Relayers allow senders to issue feeless transactions, i.e. users can send transactions from addresses that do not hold ether yet. Such relayer nodes can also easily deanonymize their users. This is especially problematic in case of non-custodial mixers, like Tornado Cash.
\subsection{Wallet and Browser Privacy}
It has been shown how online trackers and cookies can aid the deanonymization of cryptocurrency users even when their coins were mixed through the use of a mixer~\cite{Goldfeder2017WhenTC}. Many users of the Ethereum blockchain make use of a tool called MetaMask, a browser extension available in most desktop browsers. As such, for future research, it would be fascinating to analyze how the use of this extension affects the privacy of Ethereum users, even with the use of mixers. It may be possible to use the techniques presented in~\cite{Goldfeder2017WhenTC} to deanonymize users. Furthermore, as many Ethereum users also make use of mobile wallets, it may be useful to investigate how mobile phones can affect cryptocurrency users' privacy and assess the privacy guarantees of these mobile wallet providers~\cite{biryukov2019privacy}.

\subsection{Privacy of UTXO-based cryptocurrencies}
We note that the deanonymizing power of quasi-identifiers (e.g. temporal activity, wallet fingerprints etc.) is also applicable to UTXO-based cryptocurrencies. Even though in that case deanonymization is slightly more involved as one need to apply our techniques not to individual addresses but rather to clusters of UTXOs. We do foresee that more potent agencies can and will engage in such deanonymization campaigns. We believe that in practice, due to the aforementioned quasi-identifiers, also Bitcoin non-custodial mixers provide drastically less privacy and fungibility than what currently the community expects from those privacy-enhancing technologies.

\section{Conclusion} \label{sec:conclusion}
In this paper, we studied how graph representation learning, time-of-day activity and gas price based profiling can be used to link Ethereum addresses owned by the same user. The Ethereum Name Service (ENS) relations in our data set provided ground truth information to quantitatively compare and analyze the  performance of these quasi-identifiers. Our results showed that recent node embedding methods had superior performance compared to user activity based profiling techniques.

Recently, several privacy-enhancing overlays have been deployed on Ethereum, such as Tornado Cash mixers. By our measurements, their decreased usability and immature user behavior prevent them from reaching their highest attainable privacy guarantees. Evaluation on heuristically linked mixing participants showed that profiling techniques, especially novel node embedding algorithms, can significantly reduce the anonymity set sizes of the mixing parties.

Finally, we investigated an active attack scenario for Ethereum confidential transactions by repurposing the  Danaan-gift attack, originally introduced for Zcash. The estimated success probability of the attack demonstrates that users should be concerned and warned about these attacks against transaction confidentiality.

We release the collected data as well as our source code to facilitate further research\ifarxiv\footnote{\url{https://github.com/ferencberes/ethereum-privacy}}\fi.

\ifarxiv
\section*{\ackname}
We thank Daniel A. Nagy, David Hai Gootvilig, Domokos M. Kelen and Kobi Gurkan for conversations and useful suggestions. Support from Project 2018-1.2.1-NKP-00008: Exploring the Mathematical 
Foundations of Artificial Intelligence and the “Big Data—--Momentum” grant of the Hungarian Academy of Sciences.
\fi

\bibliographystyle{plain}
\bibliography{sample}

\appendix
\section{Ethereum basics}\label{sec:ethbasics}
Ethereum is a cryptocurrency built on top of a blockchain~\cite{wood2014ethereum}. There are two types of accounts in Ethereum: externally owned accounts (EOAs) and contract accounts, also known as smart contracts. The global state of the system consists of the state of all different accounts. EOAs are controlled by an asymmetric cryptographic key pair, while smart contracts are controlled by their code stored in persistent, immutable storage. EOAs can issue transactions, which might alter the global state. Transactions can either create a new contract account or call existing accounts. Accounts have balances in ether, the native currency of Ethereum, and are denominated in wei where $1ETH=10^{18}wei$.

Calls to EOAs can transfer Ether to the callee, while contract calls execute the code associated with the smart contract. The contract execution might alter the storage of the account, moreover can call to other accounts - these are called \textit{internal transactions}. Contract code is executed in the Ethereum Virtual Machine (EVM).



\subsection{Gas mechanism} \label{sec:gasmechanism}
A crucial aspect of the EVM is the gas mechanism. To every EVM opcode, there is a gas amount assigned, which is deemed to price the computational complexity of that opcode. For instance, adding two elements on top of the stack consumes only $3$ gas, but storing a non-zero stack element in the persistent storage burns \num[group-separator={,}]{20000} gas. The base gas fee for every transaction is \num[group-separator={,}]{21000} gas, which is not paid for internal transactions. Therefore, whenever one executes a smart contract code in the EVM, the execution consumes a certain amount of gas. At each transaction, the sender needs to define the maximum number of gas, called gas limit, they allow their transaction to consume. Usually, due to the dynamic nature of the state, one does not know statically how much gas would her transaction burn. If a transaction does not consume all the gas assigned to it, then surplus gas is refunded to the caller, however, if a transaction runs out of gas, then all state changes are reverted and assigned gas is taken from the caller.

As of now, gas can only be purchased by Ethereum's native currency, ether, at a dynamically changing price, called gas price. 
Miners are naturally incentivised to insert transactions with higher gas prices into their blocks to increase their collected transaction fees.

\end{document}